\documentclass[journal,comsoc]{IEEEtran}
\input{preamble}
\allowdisplaybreaks

\begin{document}  
% \begin{sloppypar}
%\begin{frontmatter}
\title{
Nonparametric Steady-State Learning for \\ Nonlinear Output Feedback Regulation
} 
\author{Shimin Wang, Martin Guay, Richard~D.~Braatz 
%\address[QUKOC]{Queen's University, Kingston, ON K7L 3N6, Canada}

\thanks{This research was supported by the U.S. Food and Drug Administration under the FDA BAA-22-00123 program, Award Number 75F40122C00200. \\
Shimin Wang and Richard D. Braatz are with Massachusetts Institute of Technology, Cambridge, MA 02142, USA (bellewsm@mit.edu, braatz@mit.edu). \\
Martin Guay is with 
Queen's University, Kingston, ON K7L 3N6, Canada (martin.guay@queensu.ca).\\
Corresponding author: Richard~Braatz}}

\maketitle
%\maketitle %\thispagestyle{empty} \pagestyle{empty}

\begin{abstract}
This article addresses the nonparametric and robust output regulation problem of the general nonlinear output feedback system with error output. 
The global robust output regulation problem for a class of general output feedback nonlinear systems with an uncertain exosystem and high relative degree can be tackled by constructing a linear generic internal model, provided that a continuous nonlinear mapping exists. 
Leveraging the proposed nonadaptive framework facilitates the conversion of the nonlinear robust output regulation problem into a robust nonadaptive stabilization formulation for the augmented system endowed with Input-to-State Stable dynamics. 
This approach removes the need for constructing a specific Lyapunov function with positive semidefinite derivatives and avoids the common assumption of linear parameterization of the nonlinear system.
The nonadaptive approach is extended by incorporating the nonparametric learning framework to ensure the feasibility of the nonlinear mapping, which can be tackled using a data-driven method. 
Moreover, the introduced nonparametric learning framework allows the controlled system to learn the dynamics of the steady-state input behavior from the signal generated from the internal model with the output error as the feedback. 
As a result, the nonparametric approach can be advantageous to guarantee the convergence of the estimation and tracking error even when the underlying controlled system dynamics are complex or poorly understood.
{The effectiveness of the theoretical results is illustrated for three practical examples: regulation of a magnetic levitation system, regulation of a virtual synchronous generator, and heading control of a surface vessel.\myr}
\end{abstract}
%\end{frontmatter}

\section{Introduction}
Output regulation is an essential control systems design problem \citep*{isidori1990output,sutton2002reinforcement} that aims to have a system track a class of desired signals while rejecting external disturbances \citep*{marconi2008uniform,isidori1990output,huang2004nonlinear}. 
Accordingly, the desired signals and external disturbances can be lumped together as exogenous signals that are generated by an autonomous differential equation called the exosystem. 
Various formulations of the output regulation problems have been investigated in the past decade, such as linear systems in \cite{francis1976internal} and nonlinear systems in \cite{marconi2008uniform} and \cite{huang2004nonlinear} with or without uncertainties in the exosystem. 
Feedforward and feedback control are widely employed generic schematics for addressing output regulation problems.
In terms of feedforward control, the results in \cite{isidori1990output} showed that the output regulation of nonlinear systems can be solved by a feedforward control synthesized from specific solvable nonlinear partial differential equations called nonlinear regulator equations.
The solvability of the nonlinear output regulation using feedforward control strictly relies on the perfect knowledge of the plant and the exosystem dynamics in the absence of uncertainties. 
%
%{\myr }

The pivotal technique employed to handle uncertainties and to solve output regulation problems in feedback control system design is the internal model principle. In this formulation, the internal model can be interpreted as an observer of the steady-state generator providing online estimates of the steady-state input.
Output regulation in linear dynamical systems was shown to be achievable by solving pole assignment problems in \cite{francis1976internal}. Since the steady-state tracking error is a linear function of the exogenous signals, the resulting linear internal model has poles placed at the poles of the original exosystem.
Regarding nonlinear output regulation, \cite{huang1994robust} revealed that the steady-state tracking error in a nonlinear system is a nonlinear function of the exogenous signals.
As a result, the feedforward control and the linear internal model become invalid in the presence of unknown parameters and nonlinearities arising from the controlled system plant and exosystem.

Multifarious versions of the internal models for various nonlinear system dynamics and diverse exosystems have been extensively provided over the last decades, such as the canonical linear internal model \citep*{nikiforov1998adaptive,serrani2001semi}, nonlinear internal models \citep*{byrnes2003limit,huang2004general}, and general or generic internal models \citep*{marconi2008uniform,wang2023nonparametric}. 
The canonical linear internal model has been successfully applied in solving heterogeneous nonlinear output regulation problems for known and uncertain linear exosystems with adaptive methods \citep*{huang2004nonlinear,basturk2015adaptive,serrani2001semi,marino2003output}, and even recently was employed to address the disturbance rejection problem of Euler–Lagrange systems  \cite{lu2019adaptive}. 
To handle nonlinear exosystems subject to more general exosystem non-sinusoidal signal classes in the absence of uncertainties, two different nonlinear internal models were introduced by \cite{byrnes2003limit} and \cite{huang2004general}. 
To remove various assumptions on the steady-state input, a generic internal model for addressing output regulation of minimum and non-minimum phase nonlinear systems was initially proposed in \cite{marconi2008uniform}.
More articles and reviews related to internal models in control theory, bioengineering, and neuroscience are cited in \cite{huang2018internal}.

In terms of the generic internal model, as pointed out by \cite{huang2018internal}, a significant advantage of \cite{marconi2008uniform} is that the generic internal model does
not rely on any specific expression of steady-state input as long as the steady-state generator exists.
In addition, the nonlinear regulator of the generic internal model ensures robust asymptotic regulation against unstructured uncertainties, as shown in \cite{bin2024robust}.
Moreover, the generic internal model can directly provide the unknown parameters arising from
the exosystem, eliminating the need for the adaptive control technique and having a significant advantage compared to the canonical linear internal model \citep*{nikiforov1998adaptive}.
In fact, the adaptive control approach faces two challenges that impede further research on the adaptive output regulation problem. 
Firstly, it requires the construction of a specific Lyapunov function for the nonlinear time-varying adaptive system with a positive semidefinite derivative to ensure the convergence of the partial state by using the LaSalle–Yoshizawa Theorem. 
This yields weaker stability, resulting in the absence of effective analysis techniques. 
Therefore, it only applies to a class of uncertain nonlinear systems in the parametric form.
Moreover, another feature of adaptive techniques requires a known and explicit regressor
determined by the controlled system structure and the exosystem, which can be found in \cite{liu2009parameter} and in the nonlinear regression case \citep*{forte2013robust}.
Secondly, when dealing with the nonlinear time-varying adaptive system in the presence of external inputs, the derivative of the Lyapunov function constructed through the adaptive method incorporates a negative semidefinite term and external inputs. 
Nevertheless, the negative semidefinite term is insufficient to guarantee the boundedness property, even with a small input term (like noise input), as illustrated using a counterexample in \cite{chen2023lasalle}.
In contrast, the nonadaptive method proposed by \cite{isidori2012robust} removes the need for constructing such Lyapunov functions for the closed-loop
system with positive semidefinite definite derivatives. It also did not require the studied nonlinear system to be in parametric form.
As a consequence, the nonadaptive method for solving output regulation problems has received considerable attention \citep*{marconi2008uniform,marconi2007output,bin2020approximate}.
Nevertheless, the generic internal model-based method for solving output regulation strictly relies on the explicit construction of a nonlinear continuous mapping function, which is %only 
assumed to exist \citep*{kreisselmeier2003nonlinear}.
Just as \cite{bin2024robust} pointed out, no general analytical
expression is known for the nonlinear regulator construction of the generic internal model-based method.
Consequently, approximation methods have been proposed using system identification least-squares techniques to select optimized parameters for the construction of nonlinear continuous mapping functions \citep*{marconi2008uniform,marconi2007output}.
It has been demonstrated that the steady-state tracking error in a nonlinear system is a nonlinear function of the exogenous signals \cite{huang1994robust}.
In recent developments, \cite{wang2023nonparametric} proposed a nonparametric learning framework for the construction of the nonlinear mapping in a general setting by establishing connections between the Generalized Sylvester Matrix Equation, the generic internal model, and a nonlinear Luenberger observer design.  This framework enables the steady-state generator to be polynomial in the exogenous signal for nonlinear output regulation while relaxing restrictive and somewhat ad hoc assumptions on the exosystem, such as the stringent requirement for an even dimension and the absence of a zero eigenvalue imposed in some results.

Based on the aforementioned statements, this article aims to address the nonlinear robust output regulation for general nonlinear output feedback systems with higher relative degrees using nonparametric learning methods. These methods are in contrast with existing techniques that are based on common adaptive control methods \citep*{tomei2023adaptive,liu2009parameter,liuzzo2007adaptive,ding2003global}.
The output regulation problem for uncertain nonlinear
systems with higher relative degree remains an active research area \citep*{tomei2023adaptive,dimanidis2020output}, with some interesting results on the design of
low-complexity, approximation-free, output-feedback controller to achieve output tracking with prescribed transient and steady-state performance \citep*{dimanidis2020output}.
The proposed nonparametric learning framework addresses the global robust output regulation problem via the construction of a linear generic internal model contingent on the existence of a continuous nonlinear mapping. This approach enables the transformation of the nonlinear robust output regulation problem into a robust nonadaptive stabilization design problem for an augmented system with Input-to-State Stable (ISS) dynamics. 
Furthermore, integrating a nonparametric learning framework ensures the viability of the nonlinear mapping, demonstrating its capability to capture intricate and nonlinear relationships without the need for the existence of a predefined model of the steady-state input behavior. 

In contrast to the methodologies in \cite{tomei2023adaptive, nikiforov1998adaptive, ding2003global} and \cite{liu2009parameter}, which rely on system dynamics to satisfy the linear-in-parameter (parameterized) condition and employ adaptive learning, the nonparametric learning framework directly learns the system dynamics of the steady-state input with a data-driven technique that requires the construction of a Hankel matrix using the internal model signal. %Consequently, it can be viewed as a kind of data-driven learning. 
Moreover, this nonparametric approach can be considered a model-free method, as it eliminates the need for a regressor, which, in adaptive methods, is typically dictated by the structure of the system and the exosystem.
The proposed methodology also overcomes the construction of a specific Lyapunov function whose derivatives can only be shown to be negative semidefinite for a restricted class of systems in adaptive methods.
In addition, this study also provides an alternative proof for Lemma 3 of  \cite{wang2023nonparametric} in terms of the solution for a time-varying equation.
In particular, oscillatory disturbances affecting the current of a repulsive magnetic levitation system, the power angle of a virtual synchronous generator, and the heading of a marine surface vessel induce unknown steady-state chattering. This behavior is significantly more complex than the constant offset disturbances addressed in earlier studies. %in \cite{uppal1974dynamic} and \cite{henson1992nonlinear}.
Overall, the nonparametric learning framework can contribute to advancing robust control strategies and learning the unknown steady-state behaviour of nonlinear systems with broader implications in many engineering applications.

The rest of this article is organized as follows.  Section \ref{section2} introduces some standard assumptions and lemmas. Section \ref{mainresults} is devoted to the presentation of the main results, which is followed by simulation examples in Section \ref{numerexam} and brief conclusions in Section~\ref{conlu}.

\textbf{Notation:} $\|\cdot\|$ is the Euclidean norm. $\emph{Id} : \mathds{R}\rightarrow \mathds{R}$ is an identity function. For $X_i\in \mathds{R}^{n_i\times m}$ with $i=1,\dots,N$, let $\col(X_1,\dots,X_N)=[X_{1}^{\!\top},\dots ,X^{\!\top}_{N}]^{\!\top}$ and $$\textnormal{diag}(X_1,\dots,X_N)=\textstyle \left[\begin{matrix}
X_1 &    & \\
& \ddots & \\
& & X_N\end{matrix}\right].$$
A function $\alpha: \mathds{R}_{\geq 0}\rightarrow \mathds{R}_{\geq 0}$ is of class $\mathcal{K}$ if it is continuous, positive definite, and strictly increasing. The notation $\mathcal{K}_o$ and $\mathcal{K}_{\infty}$ identifies the subclasses of bounded and unbounded $\mathcal{K}$ functions, respectively. For functions $f_1(\cdot)$ and $f_2(\cdot)$ with compatible dimensions, their composition $f_{1}(f_2(\cdot))$ is denoted by $f_1\circ f_2(\cdot)$. For two continuous and positive definite functions $\kappa_1(\varsigma)$ and $\kappa_2(\varsigma)$, $\kappa_1\in \mathcal{O}(\kappa_2)$ means that  $\limsup_{\varsigma\rightarrow 0^{+}}\frac{\kappa_1(\varsigma)}{\kappa_2(\varsigma)}<\infty$.

%\section{Problem Formulation and Assumptions}
\section{Problem Formulation and Preliminaries}\label{section2}
\subsection{System Model and Problem Formulation}
% Preamble required: \usepackage{booktabs}

We consider a class of nonlinear control systems modelled by
\begin{subequations}\label{second-nonlinear-systems}
\begin{align}
\dot{z}&=f(z,y,v,w),\label{second-nonlinear-systems-a}\\
\dot{x}&= A_cx+g(z,y,v,w)+B_cbu,\\
 y &=C_cx, \\
 e &= y - h(v,w),
 \end{align}
\end{subequations}
where the state vector comprises $z\in \mathds{R}^{n_z}$ and $x\in \mathds{R}^{r}$. The $z$-subsystem represents the internal dynamics, while the $x$-subsystem represents a chain of $r$ integrators ($r\geq 1$) perturbed by the nonlinear term $g(z,y,v,w)=\col(g_1,\dots, g_r)$. 
In this setup, $y\in \mathds{R}$ is the measured output, $e\in \mathds{R}$ is the tracking error, $u\in \mathds{R}$ is the control input, and $b$ is a positive constant (unknown high-frequency gain). The vector $w\in \mathds{W}\subset \mathds{R}^{n_w}$ represents uncertain parameters, where $\mathds{W}$ is a prescribed compact set containing the origin.
\begin{table}[h]
    \centering
    \renewcommand{\arraystretch}{1.3}
    \caption{System Notation and Description}
    \label{tab:system_description}
    \begin{tabular}{cl} % Removed vertical lines, simplified alignment
        \toprule % Top horizontal line
        \textbf{Symbol} & \textbf{Description} \\
        \midrule % Middle horizontal line
        $(z, x)$ & System states: $z \in \mathds{R}^{n_z}$ (nonlinear), $x \in \mathds{R}^{r}$ (linear) \\
        $r$ & Relative degree ($r \geq 1$) \\
        $y, e, u$ & Output, Tracking error, Control input (all in $\mathds{R}$) \\
        $v$ & Exogenous signal vector in $\mathds{R}^{n_v}$ \\
        $w$ & Uncertain parameter vector in $\mathds{W} \subset \mathds{R}^{n_w}$ \\
        $g(\cdot)$ & Nonlinear coupling term $\col(g_1, \dots, g_r)$ \\
        $A_c, B_c, C_c$ & Canonical form matrices \\
        $b$ & High-frequency gain ($b > 0$) \\
        \bottomrule % Bottom horizontal line
    \end{tabular}
\end{table}
The functions $h(\cdot)$, $f(\cdot)$, and $g_i(\cdot)$ are sufficiently smooth and vanish at the origin, satisfying $h(0,w)=0$, $f(0,0,0,w)=0$, and $g_i(0,0,0,w)=0$ for all $w\in \mathds{W}$. The exogenous signal $v(t)\in \mathds{R}^{n_v}$ represents the reference input and/or disturbance, generated by the exosystem:
\begin{align}\label{eqn: exosystem system}
\dot{v}=S(\sigma)v,
\end{align}
where $\sigma \in \mathds{S}\subset \mathds{R}^{n_\sigma}$ represents the exosystem uncertainty, and $S(\sigma)$ is a matrix depending on $\sigma$.
The matrices $A_c\in \mathds{R}^{r\times r}$, $B_c\in \mathds{R}^r$, and $C_c\in \mathds{R}^{1\times r}$ form a canonical representation:
\begin{align*}
    A_c=\left[\begin{matrix}\textbf{0}& I_{r-1}\\
    0& \textbf{0}
    \end{matrix}\right]\!, \quad  C_c=\col(1, \textbf{0}_{r-1})^\top, \quad B_c=\col(\textbf{0}_{r-1},1),
\end{align*}
where $I_{r-1}$ is the identity matrix and $\textbf{0}_{r-1}$ is the zero vector of dimension $r-1$. The system notation and description is summarized in Table \ref{tab:system_description}.

The nonlinear robust output regulation problem is formulated as follows:
\begin{prob} \label{Prob: second-order-Output-regulation}
Consider the nonlinear system \eqref{second-nonlinear-systems}--\eqref{eqn: exosystem system}. Given compact subsets $\mathds{S}\subset \mathds{R}^{n_{\sigma}}$, $\mathds{W}\subset \mathds{R}^{n_w}$, and $\mathds{V}\subset \mathds{R}^{n_v}$ (with $\mathds{W}$ and $\mathds{V}$ containing the origin), design a control law such that, for all initial conditions $v(0)\in \mathds{V}$, $\sigma \in \mathds{S}$, $w\in \mathds{W}$, and any initial state $\textnormal{\col}(z(0), x(0))\in\mathds{R}^{n_z+r}$, the solution of the closed-loop system exists and remains bounded for all $t\geq 0$, and the tracking error satisfies $\lim\limits_{t\rightarrow\infty}e(t)=0$.
\end{prob}

\subsection{Assumptions and Coordinate Transformation}
We first state the standard assumptions regarding the exosystem and the regulator equations.

\begin{ass}\label{ass0}
For all $\sigma \in \mathds{S}$, the eigenvalues of the exosystem matrix $S(\sigma)$ are distinct and lie on the imaginary axis.
\end{ass}

\begin{ass}\label{ass1}
There exist globally defined smooth functions $\bm{x}(v,w,\sigma)$, $\bm{z}(v,w,\sigma)$, and $\bm{u}(v,w,\sigma)$ solving the regulator equations:
\begin{subequations}\label{regulator-1}
\begin{align}
\frac{\partial \bm{z}(\mu)}{\partial v} S(\sigma)v &= f(\bm{z}(\mu), h(v,w), v, w), \\
\frac{\partial\bm{x}(\mu)}{\partial v}S(\sigma)v &=  A_c\bm{x}(\mu)+ g(\bm{z}(\mu),\bm{x}_1(\mu),\mu)+B_cb\bm{u}(\mu),
\end{align}
\end{subequations}
subject to $\bm{z}(0,w,\sigma)=0$ and $\bm{x}_1(\mu)=h(v,w)$, where $\mu=\col(v,w,\sigma)$.
\end{ass}

\begin{ass}[Minimum-phase condition]\label{H2}
The inverse dynamics, defined relative to the steady-state manifold $\bm{z}(\mu)$, given by
\begin{align}\label{Reg-2}
\dot{\bar{z}} &= f(\bar{z}+\bm{z}(\mu),e+h(v,w),v, w) - f(\bm{z}(\mu), h(v,w),v, w),
\end{align}
are input-to-state stable (ISS) with state $\bar{z}=z-\bm{z}(\mu)$ and input $e$.
\end{ass}

\begin{rem}\label{rem-barV}
As a consequence of Assumption \ref{H2} and the changing supply function technique \citep{sontag1995changing}, there exists a smooth ISS-Lyapunov function $\bar{V}_{\bar{z}}(\bar{z})$ satisfying
$$ \underline{\alpha}_{\bar{z}}(\|\bar{z}\|)\leq \bar{V}_{\bar{z}}(\bar{z})\leq \bar{\alpha}_{\bar{z}}(\|\bar{z}\|), $$
such that its time derivative along \eqref{Reg-2} satisfies
$$\dot{\bar{V}}_{\bar{z}}(\bar{z})\leq -\Delta_{\bar{z}}(\bar{z})\|\bar{z}\|^2+\delta_{\bar{z}} e^2,$$
where $\Delta_{\bar{z}}(\cdot)$ is a positive smooth function and $\delta_{\bar{z}} > 0$ is a constant.
\end{rem}

To facilitate the control design for the system with relative degree $r \geq 2$, we employ the input-driven filter as in \cite{jiang2004unifying},
\begin{align}\label{input-filter}
    \dot{\hat{x}}=A\hat{x}+B_cu,
\end{align}
where $A=A_c-\lambda C_c$ is Hurwitz with $\lambda=\col(\lambda_1,\dots, \lambda_r)$. Applying the coordinate transformation $\tilde{x}_i=b^{-1}x_i-\hat{x}_i$ ($i=1,\dots,r$) yields the transformed system:
\begin{subequations}\label{nonlinear-systems-trans}
\begin{align}
\dot{z}&=f(z,y,v,w),\\
\dot{\tilde{x}}&= A\tilde{x}+b^{-1}(\lambda y+g(z,y,v,w)),\\
 \dot{y} &=b\hat{x}_2+b\tilde{x}_2+g_1(z,y,v,w), \\
  \dot{\hat{x}}_i&= \hat{x}_{i+1}-\lambda_i\hat{x}_1,\quad i=2,\dots,r-1,\\
  \dot{\hat{x}}_r&=u-\lambda_r\hat{x}_1.
 \end{align}
\end{subequations}

Ideally, the aim is for the filter state $\hat{x}(t)$ to track a steady-state profile induced by the ideal control $\bm{u}(v,w,\sigma)$.
\begin{ass}\label{ass0i}
The ideal control law $\bm{u}(v,w,\sigma)$ defined in Assumption \ref{ass1} is a polynomial in $v$ with coefficients dependent on $w$ and $\sigma$.
\end{ass}

The following lemma establishes the structure of the steady-state generator for the filter states, which is crucial for the internal model design.

\begin{lem}[Steady-State Filter Dynamics]\label{lem:steady_state_poly}
Under Assumptions \ref{ass1} and \ref{ass0i}, the steady-state control $\bm{u}(\mu)$ can be expressed as $\bm{u}(\mu) = \Gamma_u(\mu) \tau_u(v)$, where $\tau_u(v)$ is a vector containing monomials of $v$ satisfying the linear dynamics $\dot{\tau}_u = \Phi_u(\sigma) \tau_u$.
Consequently, there exists a unique steady-state trajectory $\bm{\hat{x}}(\mu)$ for the filter \eqref{input-filter}, defined by $\bm{\hat{x}}(\mu) = P_u(\mu)\tau_u(v)$, where $P_u$ is the unique solution to the Sylvester equation:
\begin{equation}
P_u \Phi_u(\sigma) = A P_u + B_c \Gamma_u(\mu).
\end{equation}
Moreover, $\bm{\hat{x}}(\mu)$ satisfies the autonomous differential equation:
\begin{align}\label{regulator-2}
    \dot{\bm{\hat{x}}}(\mu) = A \bm{\hat{x}}(\mu)+ B_c \bm{u}(\mu).
\end{align}
\end{lem}

\begin{proof}
Since $S(\sigma)$ has eigenvalues on the imaginary axis (Assumption \ref{ass0}), the generator matrix $\Phi_u(\sigma)$ for the polynomial basis $\tau_u(v)$ also shares this spectral property. Given that $A$ is Hurwitz, the spectra of $A$ and $\Phi_u(\sigma)$ are disjoint. Thus, the Sylvester equation $P_u \Phi_u - A P_u = B_c \Gamma_u$ has a unique solution $P_u$. Differentiating $\bm{\hat{x}} = P_u \tau_u$ yields \begin{align*}\dot{\bm{\hat{x}}} &= P_u \Phi_u \tau_u\\
&= (A P_u + B_c \Gamma_u)\tau_u \\
&= A \bm{\hat{x}} + B_c \bm{u}.\end{align*}
This completes the proof. 
\end{proof}
Defining the error $\bm{E}(\mu)=b^{-1}\bm{x}(\mu)-\bm{\hat{x}}(\mu)$, the solution to the regulator equations for the composite system is $$\{\bm{z}(\mu), \bm{E}(\mu), \bm{y}(\mu), \bm{\hat{x}}(\mu), \bm{u}(\mu)\}.$$

\begin{rem}\label{remPE}
Let $\bm{\hat{x}}_2(\mu)$ denote the second component of $\bm{\hat{x}}(\mu)$. Under Assumptions \ref{ass0}--\ref{ass0i}, $\bm{\hat{x}}_2$ admits a finite trigonometric representation \citep{liu2009parameter}:
\begin{align}\label{remPE-trsin}
\bm{\hat{x}}_2(v(t),\sigma,w)
    = \sum\nolimits_{j=1}^{n} C_{j}(v(0),w,\sigma)\, e^{\imath \hat{\omega}_{j} t},
\end{align}
where $\hat{\omega}_j$ are distinct real frequencies.
\end{rem}

\begin{ass}\label{ass5-explicit}
For any initial condition $v(0)\in \mathds{V}$ and parameters $w\in \mathds{W}$, $\sigma\in \mathds{S}$, the coefficients satisfy $C_{j}(v(0), w,\sigma)\neq 0$ for all $1\leq j \leq n$.
\end{ass}

\subsection{Generic Internal Model Design}
Under Assumptions \ref{ass0} and \ref{ass0i}, there exists a positive integer $n$ such that the steady-state signal $\bm{\hat{x}}_2(\mu)$ satisfies the autonomous differential equation:
\begin{align} \label{aode-explicit}
\frac{d^{n}\bm{\hat{x}}_2}{dt^{n}}+a_{n}(\sigma)\frac{d^{n-1}\bm{\hat{x}}_2}{dt^{n-1}}+\dots+a_{1}(\sigma)\bm{\hat{x}}_2=0,
\end{align}
where $a_{i}(\sigma) \in \mathds{R}$. By defining the state vector $\bm{\xi} = \col(\bm{\hat{x}}_2, \dot{\bm{\hat{x}}}_2, \dots, \bm{\hat{x}}_2^{(n-1)})$, \eqref{aode-explicit} can be represented in state-space form, referred to as the \textit{steady-state generator},
\begin{subequations}\label{stagerator}
\begin{align}
\dot{\bm{\xi}} &= \Phi(a(\sigma)) \bm{\xi}, \\
\bm{\hat{x}}_2 &= \Gamma \bm{\xi},
\end{align}
\end{subequations}
where the companion matrix $\Phi(a(\sigma))$ and the output vector $\Gamma$ are defined as
\begin{align*}
  \Phi(a(\sigma)) &= \left[
                      \begin{array}{c|c}
                        \textbf{0}_{(n-1)\times 1} & I_{n-1} \\
                        \hline
                        -a_{1}(\sigma) & -a_{2}(\sigma),\dots,-a_{n}(\sigma) \\
                      \end{array}
                    \right],\\
  \Gamma &= \left[ 1, 0, \dots, 0 \right]_{1\times n}.
\end{align*}

To design an internal model that reproduces $\bm{\hat{x}}_2$, introduce a dynamic system with state $\eta \in \mathds{R}^{2n}$,
\begin{align}\label{explicit-mas1}
\dot{\eta} = M\eta + N\hat{x}_2,
\end{align}
where the pair $(M, N)$ is chosen in the controllable canonical form:
\begin{subequations}\label{MNINter}
\begin{align}
M&= \left[
      \begin{array}{c|c}
        \textbf{0}_{(2n-1)\times 1} & I_{2n-1} \\
        \hline
        -m_{1} & -m_{2},\dots,-m_{2n} \\
      \end{array}
    \right], \\
N&= \left[ 0, \dots, 0, 1 \right]_{1\times 2n}^{\top}.
\end{align}
\end{subequations}
The coefficients $m_j$ are selected such that the matrix $M$ is Hurwitz. Let $P_M(\lambda) = \lambda^{2n} + \sum_{j=1}^{2n} m_j \lambda^{j-1}$ denote the characteristic polynomial of $M$. Since the eigenvalues of $\Phi(a)$ are distinct with zero real parts (by Assumption \ref{ass0i}) and $M$ is Hurwitz, their spectra are disjoint. This condition ensures nonsingularity of the matrix polynomial function:
$$ \Xi(a) = \underbrace{\Phi(a)^{2n}+\sum\nolimits_{j=1}^{2n}m_{j}\Phi(a)^{j-1}}_{P_M(\Phi(a))} \in \mathds{R}^{n \times n }. $$

\subsubsection*{The Generalized Sylvester Equation}
To immerse the generator dynamics \eqref{stagerator} into the internal model \eqref{explicit-mas1}, we seek a transformation matrix $Q \in \mathds{R}^{2n \times n}$ satisfying the \textit{Generalized Sylvester Matrix Equation} (see \cite{wang2023nonparametric}):
\begin{align}
M Q - Q \Phi(a(\sigma)) = -N\Gamma. \label{MNGAMMAPhi}
\end{align}
Due to the specific canonical structure of $(M, N)$, the unique solution $Q$ admits an explicit form. Let $Q$ be partitioned as $Q = \col(Q_1, \dots, Q_{2n})$ where $Q_j \in \mathds{R}^{1 \times n}$.
Using the commutativity property $\Phi(a)\Xi(a)^{-1}=\Xi(a)^{-1}\Phi(a)$ and the structural identity $\col(\Gamma, \Gamma \Phi(a),\dots, \Gamma\Phi(a)^{n-1} )=I_n$, the upper block of $Q$ can be explicitly derived as
\begin{align}
\col & (Q_1(a),\dots,Q_n(a))\nonumber\\
=&\ \col(\Gamma \Xi(a)^{-1},\Gamma \Phi(a)\Xi(a)^{-1},\dots, \Gamma \Phi(a)^{n-1}\Xi(a)^{-1})\nonumber\\
=&\ \underbrace{\col(\Gamma,\Gamma \Phi(a),\dots, \Gamma \Phi(a)^{n-1})}_{I_n} \Xi(a)^{-1}, \label{XIQA-explicit}
\end{align}
with the general term given by $Q_{j}(a)=\Gamma \Xi(a)^{-1}\Phi(a)^{j-1}$.

Finally, defining the steady-state internal model coordinate $\bm{\eta}^{\star} = Q \bm{\xi}$, we construct the Hankel matrix $\Theta(\bm{\eta}^{\star})$ as in \cite{afri2016state}:
\begin{align}
\Theta (\bm{\eta}^{\star}) \equiv \left[\begin{matrix}\eta^{\star}_{1} &\dots&\eta^{\star}_{n}\\
\vdots&\ddots&\vdots\\
\eta^{\star}_{n} &\dots&\eta^{\star}_{2n -1}
\end{matrix}\right] \in \mathds{R}^{n \times n}.
\end{align}
As shown in Lemma 3 of \cite{wang2023nonparametric}, substituting the solution of \eqref{MNGAMMAPhi} into the integral form of the internal model ensures that $\bm{\eta}^{\star}$ satisfies the target dynamics $\dot{\bm{\eta}}^{\star} = M\bm{\eta}^{\star} + N\bm{\hat{x}}_2$, completing the design.

\subsection{Error Dynamics}

Perform coordinate and input transformations on the composite systems \eqref{eqn: exosystem system}, \eqref{nonlinear-systems-trans}, and \eqref{explicit-mas1} to give
\begin{align*}
\bar{z}&= z-\bm{z},&  \bar{x}= \tilde{x}-\bm{E}, \\
\bar{\eta}&= \eta-\bm{\eta}^{\star}-Nb^{-1}e,& e=y-\bm{x}_1, 
\end{align*}
which yields an error system in the form:
\begin{subequations}\label{Main-sys1}\begin{align}
\dot{\bar{z}} %=& f(z,y,v,w)-f(\bm{z},h(v,w),v,w)\nonumber\\
%=&f(\bar{z}+\bm{z},e+h(v,w),v,w)-f(\bm{z},h(v,w),v,w)\nonumber\\
&= \bar{f}(\bar{z},e,\mu),\\
\dot{\bar{x}}%= &\dot{\bm{E}}-\dot{\tilde{x}}\nonumber\\
%= &b^{-1}\dot{\bm{x}}-\dot{\bm{\hat{x}}}-\dot{\tilde{x}}\nonumber\\
%= &b^{-1}A_c\bm{x}+ b^{-1}g(\bm{z},\bm{x}_1,\mu)+B_c\bm{u}\nonumber\\
%&- A\bm{\hat{x}}-B_c\bm{u}- A\tilde{x}-b^{-1}(\lambda y+g(z,y,v,w))\nonumber\\
%= &b^{-1}A\bm{x}+b^{-1}\lambda \bm{x}_1+ b^{-1}g(\bm{z},\bm{x}_1,\mu)\nonumber\\
%&- A\bm{\hat{x}}- A\tilde{x}-b^{-1}(\lambda y+g(z,y,v,w))\nonumber\\
%= &A\big[b^{-1}\bm{x}-\bm{\hat{x}}-\tilde{x}\big]+b^{-1}\lambda( \bm{x}_1-y)\nonumber\\
%&+ b^{-1}\big(g(\bm{z},\bm{x}_1,\mu)-g(z,y,v,w)\big)\nonumber\\
%= &A\big[\bm{E}-\tilde{x}\big]-b^{-1}\lambda e\nonumber\\
%&+ b^{-1}\big(g(\bm{z},\bm{x}_1,\mu)-g(\bar{z}+\bm{z},e+\bm{x}_1,\mu)\big)\nonumber\\
&=A\bar{x}+ b^{-1}\!\big[\bar{g}(\bar{z}, e, \mu)+\lambda e\big],\\
\dot{\bar{\eta}}%=&\dot{\eta}-\dot{\bm{\eta}}^{\star}-Nb^{-1}\dot{e}\nonumber\\
%=&M\eta+N\hat{x}_2-M\bm{\eta}^{\star} - N\bm{\hat{x}}_2\nonumber\\
%&-N\hat{x}_2-N\bar{x}_2 +N\bm{\hat{x}}_2-Nb^{-1}\bar{g}_1(\bar{z}, e, \mu)\nonumber\\
%=&M\eta-M\bm{\eta}^{\star} -N\bar{x}_2 -Nb^{-1}\bar{g}_1(\bar{z}, e, \mu) \nonumber\\
&=M\bar{\eta} -N\!\left(\bar{x}_2 -b^{-1}e+b^{-1}\bar{g}_1(\bar{z}, e, \mu)\right)\!,\\
\dot{e}%=&b\hat{x}_2+b\tilde{x}_2+g_1(z,y,v,w)- \dot{\bm{x}}_1\nonumber \\
%=&b\hat{x}_2+b\tilde{x}_2+g_1(z,y,v,w)\nonumber\\
%&-   C_cA_c\bm{x}(\mu)-C_cg(\bm{z}(\mu),\bm{x}_1(\mu),\mu)-C_cB_cb\bm{u}(\mu)\nonumber\\
%=&b\hat{x}_2+b\tilde{x}_2-   \bm{x}_2(\mu)\nonumber\\
%&+g_1(z,y,v,w)-g_1(\bm{z}(\mu),\bm{x}_1(\mu),\mu)\nonumber\\
%=&b\hat{x}_2+b\tilde{x}_2-  b\bm{E}_2 -b\bm{\hat{x}}_2\nonumber\\
%&+g_1(\bar{z}+\bm{z},e+\bm{x}_1,v,w)-g_1(\bm{z},\bm{x}_1,\mu)\nonumber\\
%=&b\hat{x}_2+b\bar{x}_2 -b\bm{\hat{x}}_2+\bar{g}_1(\bar{z}, e, \mu) \nonumber\\
&=b(\hat{x}_2-\bm{\hat{x}}_2)+b\bar{x}_2 +\bar{g}_1(\bar{z}, e, \mu) ,\\
 \dot{\hat{x}}_i
 &= \hat{x}_{i+1}-\lambda_i\hat{x}_1,\quad i=2,\dots,r-1,\\
  \dot{\hat{x}}_r
  &= u-\lambda_r\hat{x}_1,
\end{align}\end{subequations}
where $\mu=\col(\sigma,v,w)$,
\begin{align*}
\bar{f}(\bar{z},e,\mu)
&=f(\bar{z}+\bm{z},e+\bm{x}_1,\mu)-f(\bm{z},\bm{x}_1,\mu),\\
\bar{g}(\bar{z}, e, \mu)
&= g(\bar{z}+\bm{z},e+\bm{x}_1,\mu)-g(\bm{z},\bm{x}_1,\mu).
\end{align*}
It can be verified that, for all $\mu\in \mathds{V}\times\mathds{W} \times\mathds{S}$, $\bar{f}(0,0,\mu)=0$ and $\bar{g}(0, 0, \mu)=0$, Problem \ref{Prob: second-order-Output-regulation} can be solved if a control law can be found to stabilize the system \eqref{Main-sys1}.

Let $\bar{x}_c=\col(\bar{\eta}, \bar{x})$ and $\bar{G}_c(\bar{z}, e, \mu)=b^{-1}\col\big(Ne-N\bar{g}_1(\bar{z}, e, \mu), \bar{g}(\bar{z}, e, \mu)+\lambda e\big)$; the system \eqref{Main-sys1} can be rewritten into the form
\begin{subequations}\label{Main-sys-com}\begin{align}
\dot{\bar{z}} 
&=\bar{f}(\bar{z},e,\mu),\\
\dot{\bar{x}}_c
&= \underbrace{\left[\begin{matrix}M & -NC_cA_c\\
\textbf{0}&A
\end{matrix}\right]}_{M_c}\bar{x}_c+\bar{G}_c(\bar{z}, e, \mu) ,\label{Main-sys-com-xc}\\
\dot{e}
&=b(\hat{x}_2-\chi(\bm{\eta}^*))+b\bar{x}_2 +\bar{g}_1(\bar{z}, e, \mu) ,\\
 \dot{\hat{x}}_i
 &= \hat{x}_{i+1}-\lambda_i\hat{x}_1, \quad i=2,\dots,r-1,\\
\dot{\hat{x}}_r
&=u-\lambda_r \hat{x}_1.
\end{align}\end{subequations}
%where $M_c=\left[\begin{matrix}M & -NC_cA_c\\
%0&A
%\end{matrix}\right]$. 
It can be verified that, for all $\mu\in \mathds{V}\times\mathds{W} \times\mathds{S}$,  $\bar{G}_c(0, 0, \mu)=\bm{0}$ and the matrix $M_c$ is Hurwitz. Hence, the $(\bar{z}, \bar{x}_c)$-subsystem in system \eqref{Main-sys-com} is in a similar form as the system (8) of \cite{wang2023nonparametric}. 
As a result,  the $(\bar{z}, \bar{x}_c)$-subsystem in system \eqref{Main-sys-com}, under Assumptions \ref{ass0}, \ref{ass1}, and \ref{H2}, admits the following properties (see Properties 1 and 2 in \cite{wang2023nonparametric}): 
%\begin{proposition}\label{Prop-1}%
\begin{prop}\label{property1}%
There exists a smooth input-to-state Lyapunov function $V_0\equiv V_0(\bar{z},\bar{x}_c)$ satisfying
\begin{align}
\underline{\alpha}_0(\|\bar{Z}\|)&\leq V_0(\bar{Z})\leq \bar{\alpha}_0(\|\bar{Z}\|),\notag\\
\dot{V}_0 &\leq  -\|\bar{Z}\|^2+\bar{\gamma}^*\bar{\gamma} \left(e\right),\label{V0}
\end{align}
for some positive constant $\bar{\gamma}^*$ and comparison functions $\underline{\alpha}_0(\cdot)\in \mathcal{K}_{\infty}$, $\bar{\alpha}_0(\cdot)\in \mathcal{K}_{\infty}$, and $\bar{\gamma}(\cdot)\in \mathcal{K}_\infty$ with $\bar{Z}=\textnormal{\col}(\bar{z},\bar{x}_c)$.
\end{prop}

\begin{prop}\label{property2}%
There are positive smooth functions $\gamma_{g0}(\cdot)$ and $\gamma_{g1}(\cdot)$ such that
$$b^2\bar{x}_2^2+\|\bar{g}_1(\bar{z}, e, \mu)\|^2\leq \gamma_{g0}(\bar{Z})\|\bar{Z}\|^2+e^2\gamma_{g1}(e).$$
\end{prop}
%\end{proposition}

\begin{rem}\label{rem:proof_prop1}
Proposition \ref{property1} can be established via Lyapunov analysis. Since $M_c$ is Hurwitz, there exists a positive definite matrix $P_c$ such that $P_c M_c + M_c^\top P_c = -2I$. Since $\bar{G}_c$ is smooth and vanishes at the origin, Lemma 7.8 in \cite{huang2004nonlinear} implies that
$$ \|P_c\bar{G}_c(\bar{z}, e ,v)\|^2 \leq \pi_1(\bar{z})\|\bar{z}\|^2 + \phi_1(e)e^2, $$
for some known smooth functions $\pi_1(\cdot) \geq 1$ and $\phi_1(\cdot) \geq 1$.
Additionally, by Remark \ref{rem-barV}, the $\bar{z}$-subsystem admits an ISS Lyapunov function $\bar{V}_{\bar{z}}$ satisfying $\dot{\bar{V}}_{\bar{z}} \leq -\Delta_{\bar{z}}(\bar{z})\|\bar{z}\|^2 + \delta_{\bar{z}}\gamma_{\bar{z}}(e)e^2$.
We construct the composite Lyapunov function $V_0(\bar{Z}) = \bar{V}_{\bar{z}}(\bar{z}) + \bar{x}_c^{\top}P_c\bar{x}_c$. Its time derivative along the trajectories of \eqref{Main-sys-com} satisfies
\begin{align*}
    \dot{V}_0 &\leq -\Delta_{\bar{z}}\|\bar{z}\|^2 + \delta_{\bar{z}}\gamma_{\bar{z}}e^2 - \|\bar{x}_c\|^2 + \|P_c\bar{G}_c\|^2 \\
    &\leq -\big(\Delta_{\bar{z}}(\bar{z}) - \pi_1(\bar{z})\big)\|\bar{z}\|^2 - \|\bar{x}_c\|^2 + \big(\delta_{\bar{z}}\gamma_{\bar{z}}(e) + \phi_1(e)\big)e^2.
\end{align*}
Choosing the design freedom in Remark \ref{rem-barV} such that $\Delta_{\bar{z}}(\bar{z}) > \pi_1(\bar{z}) + 1$ implies that $\dot{V}_0 \leq -\|\bar{z}\|^2 - \|\bar{x}_c\|^2 + \bar{\gamma}^*\bar{\gamma}(e)$, which confirms \eqref{V0}.
\end{rem}

\section{Main Results}\label{mainresults}
\subsection{Nonparametric Learning in Robust Output Regulation}
\begin{figure}[htp]
    \centering
\tikzstyle{block} = [draw, fill=mitSilverGray!10, rectangle,
    minimum height=1.5em, minimum width=6em]
    \tikzstyle{controllaw} = [draw, fill=mitBlue!60, rectangle,
    minimum height=1.5em, minimum width=1em]
      \tikzstyle{controllaw2} = [draw, fill=mitBlue!60, rectangle,
    minimum height=1.5em, minimum width=2em]
    \tikzstyle{iterative} = [draw, fill=mitRed!60, rectangle,
    minimum height=1.5em, minimum width=2em]
    \tikzstyle{iterative2} = [draw=black, fill=mitGreen!60, rectangle,
    minimum height=1.5em, minimum width=2em]
        \tikzstyle{iterative3} = [draw=black, fill=mitGreen!60, rectangle,
    minimum height=0.5em, minimum width=2em]
            \tikzstyle{iterative4} = [draw=black, fill=mitPink!60, rectangle,
    minimum height=0.5em, minimum width=2em]
    \tikzstyle{interModel} = [draw, fill=mitPurple!60, rectangle,
    minimum height=1.5em, minimum width=6em]
       \tikzstyle{iterative5} = [draw, fill=mitDPink!60, rectangle,
    minimum height=1.5em, minimum width=6em]
       \tikzstyle{iterativeframe} = [draw=black!40, dashed, fill=mitDPink!30, rectangle, rounded corners,
    minimum height=4.75em, minimum width=20em]
    \tikzstyle{iterative5frame} = [draw, fill=mitPink!60, rectangle,
    minimum height=1.5em, minimum width=2em]
    
\tikzstyle{sum} = [draw, fill=green!20, circle, node distance=1cm,minimum size=0.05cm]
\tikzstyle{input} = [coordinate]
\tikzstyle{output} = [coordinate]
\tikzstyle{pinstyle} = [pin edge={to-,thin,black}]
% The block diagram code is probably more verbose than necessary
\begin{tikzpicture}[scale=0.75,every node/.style={transform shape}, node distance=2cm,>=latex']
    % We start by placing the blocks
    \node [input, name=input] {};
    \node [sum, fill=mitLightBlue, right =1.625cm of input] (sum) {};
    \node [controllaw, right = 1cm of sum] (controller) {~$u = \alpha_{s,r}(\epsilon_1,\epsilon_2,\dots,\epsilon_r,k^*,\eta,\hat{a})$~};
     \node [controllaw2, below left = 4.9cm and 0.195cm of sum] (controller-ttle) {Controller};
    
     \node [interModel, right = 0.235cm of controller-ttle] (interModel-ttle) {Internal Model};
     \node [iterative, right  = 0.235cm of interModel-ttle] (iterative-ttle) {Input-driven Filter};
     \node [iterative2, right = 0.235cm of iterative-ttle] (recursive-ttle) {Recursive Equations};
    \node [block, right = 1cm of controller,
            node distance=3cm, pin={[pinstyle]below:Disturbances}, fill=mitSilverGray!60] (system) {System};
    % We draw an edge between the controller and system block to
    % calculate the coordinate u. We need it to place the measurement block.
    \draw [->] (controller) -- node[name=u,above] {$u$} (system);
  %  \node [output, right =0.3cm of system] (output) {};
   %  \node [output, right=0.6cm of output] (output1) {};
    \node [output,above =0.5cm of controller] (measurements) {};
    \node [output, right =0.1cm of measurements] (measurements2) {};
   
     \node [iterative, below left =0.85cm and -4.95cm of controller] (Ivlearn) { $\dot{\hat{x}}=A\hat{x}+ B_cu$};
     \node [output,left =1.6cm of Ivlearn] (Ivlearn1) {};
      \node [interModel, below =0.5cm of Ivlearn] (IMmodel) {$\dot{\eta}=M\eta+N\hat{x}_2$};
    %\node [iterativeframe, below right =0.265cm and -3.6000 of IMmodel] (iterativeframe) { };
    \node [output,left =0.35cm of IMmodel] (IMmodel2) {};
    \node [output,left =1.175cm of IMmodel2] (IMmodel3) {};
    
   \node [iterative3, below=0.65cm of sum] (ErrorTran) {$\begin{matrix}\alpha_{s,1}( \epsilon_1,k^*,\eta,\hat{a})\\
     \alpha_{s,2}(\epsilon_1,\epsilon_2,k^*,\eta,\hat{a},\hat{x}_1)\\
      \vdots\\
      \alpha_{s,r}(\epsilon_1,\dots,\epsilon_r,k^*,\eta,\hat{a},\hat{x}_1)
      \end{matrix}$};
       \node [iterativeframe, below right =0.359cm and -3.8000 of IMmodel] (iterativeframe) { };
       \node [iterative4, below right=0.5cm and -3.65cm of IMmodel] (Nonparametric) {$\dot{\hat{a}}  =- k_a \Theta(\eta)^{\!\top}\left[\Theta (\eta)\hat{a} +\textnormal{\col}(\eta_{n +1},\dots{},\eta_{2n})\right]$};
       \node [output, below = 0.5cm of ErrorTran.south east] (Nonparametric-a) {};
       \node [iterative5,below =0.15cm of Nonparametric] (Nonparametric2) {$\Theta (\bm{\eta} ^{\star} )a +\col(\bm{\eta}_{n +1}^{\star},\dots,\bm{\eta} ^{\star}_{2n})=\bm{0}$};
       \node [iterative5, below left =-0.45cm and 1.595cm of Nonparametric2] (Nonparametric3) {Time-varying Equation};
     \node [iterative5frame, above left =0.15cm and -3.585cm of Nonparametric3] (Nonparametric-name) {Nonparametric learning};
       %iterative5frame
       \node [below =0.5cm of IMmodel] (Nonparametric4){};
      \node [right=1.5cm of ErrorTran.north east] (ErrorTran1){};
    % Once the nodes are placed, connecting them is easy.
    \draw [draw,thick,->] (input) -- node[above] {$h(v,w)$} (sum);
   % \draw [->] (sum) --  (controller);
    %\draw [->] (system) -- (output1);
    \draw [thick,->] (system.north)|-(measurements2)--(measurements)node[below] {$y$}-| node[pos=0.95,left] {$-$}(sum);
    \draw [thick,->] (u)|-(Ivlearn);
    \draw [thick,->] (IMmodel)--node[pos=0.5, below] {$\eta$}($(ErrorTran.east)+(0,-0.6125)$);%-|(5.225-0.25,-0.3);%(Ivlearn1)
    \draw [thick,->] (Ivlearn)--node[pos=0.5, right] {$\hat{x}_2$}(IMmodel);
    \draw [thick,->] (Ivlearn)--node[pos=0.5, below] {$\hat{x}$}($(ErrorTran.east)+(0,0.45)$);
    \draw [thick,->] (sum)--node[pos=0.5,right]{$\epsilon_1=e$}(ErrorTran);
    \draw [thick,->] (ErrorTran.north east)--node[pos=0.5,above]{$\alpha_{s,r}$}(ErrorTran1)-|(controller.south);
    \draw [thick,->] (IMmodel)--node[pos=0.5,right]{$\eta$}(Nonparametric4);
    \draw [thick,->] (Nonparametric)-|(Nonparametric-a) --node[pos=0.25,left]{$\hat{a}$}(ErrorTran.south east);
\end{tikzpicture}    
\caption{Nonparametric learning in robust output regulation.}\label{framework2}
\end{figure}
We now provide an alternative proof of \citep*[Lemma 3]{wang2023nonparametric} in terms of the time-varying equation 
\begin{align}\label{a-explicit0}
\Theta (\bm{\eta} ^{\star} )a +\col(\bm{\eta}_{n +1}^{\star},\dots,\bm{\eta} ^{\star}_{2n})=\bm{0},
\end{align}
where $\bm{\eta} ^{\star}=\underbrace{Q \bm{\xi}}_{\theta}$.
\begin{lem}\label{Lemmamappingf-2} Under Assumptions \ref{ass0}--\ref{ass5-explicit}, the Hankel real matrix $\Theta (\bm{\eta} ^{\star} )$ is nonsingular and the linear time-varying equation \eqref{a-explicit0} has a unique solution $\check{a} (\theta(t))=a$ for all $t\geq 0$.
\end{lem}
\begin{proof}   
From $Q = \col(Q_1, \dots, Q_{2n})$ and $\theta =Q \bm{\xi} $, $Q_{j}(a)=\Gamma \Xi(a)^{-1}\Phi(a )^{j-1} \in \mathds{R}^{1\times n }$, $1\leq j\leq 2n $, the real Hankel matrix is given by
\begin{align}%\label{theta-seperate-explicit}
\Theta (\theta)=& \left[\begin{matrix}Q_1 \bm{\xi}&Q_2 \bm{\xi}&\cdots&Q_n \bm{\xi}\\
Q_2 \bm{\xi}&Q_3 \bm{\xi}&\cdots&Q_{n+1} \bm{\xi}\\
\vdots&\vdots&\ddots&\vdots\\
Q_{n}\bm{\xi} &Q_{n +1}\bm{\xi}&\cdots&Q_{2n -1}\bm{\xi}
\end{matrix}\right]\nonumber\\
=& \left[\begin{matrix}Q_1 \bm{\xi}&Q_1\Phi(a)  \bm{\xi}&\cdots&Q_1 \Phi(a)^{n -1}\bm{\xi}\\
Q_2 \bm{\xi}&Q_2\Phi(a) \bm{\xi}&\cdots&Q_{2} \Phi(a)^{n -1}\bm{\xi}\\
\vdots&\vdots&\ddots&\vdots\\
Q_{n}\bm{\xi} &Q_{n}\Phi(a)\bm{\xi}&\cdots&Q_{n}\Phi(a)^{n -1}\bm{\xi}
\end{matrix}\right]\nonumber\\
=&\ \underbrace{\col(Q_1,\dots,Q_n)}_{\Xi(a)^{-1}}  \underbrace{\left[\begin{matrix}\bm{\xi}  &\Phi(a) \bm{\xi} &\dots & \Phi(a)^{n -1} \bm{\xi} \end{matrix}\right]}_{\Pi}.\nonumber
\end{align}
where the columns of the Krylov matrix $\Pi$ form the Krylov subspace.
     Under Assumption \ref{ass0}, the matrix $\Phi(a)$ is diagonalizable with distinct eigenvalues $\lambda_1=\imath \hat{\omega}_{1},\dots,\lambda_n=\imath \hat{\omega}_{n}$. 
    Moreover, the matrix $\Phi(a)$ is in companion form. 
    Therefore, from \cite{kalman1984generalized} and \cite{neagoe1996inversion}, there exists  a diagonalizable matrix $\Lambda$ with $\Lambda=\textnormal{diag}(\lambda_1,\dots,\lambda_n)$ and a nonsingular Vandermonde matrix $$P_{\Lambda}=\left[\begin{matrix}1  & 1 & \dots &1\\
    %1  &  & \cdots &\lambda_1^{n}\\
    \lambda_1  & \lambda_2 & \cdots& \lambda_n\\
    \vdots &\vdots& \ddots & \vdots \\
    \lambda_1^{n-1}  & \lambda_2^{n-1} & \cdots& \lambda_n^{n-1}\\
    \end{matrix}\right]$$  such that 
    $\Phi(a)= P_{\Lambda}\Lambda P_{\Lambda}^{-1}$.
As a result, from \eqref{stagerator}, let $\nu(t)=P_{\Lambda}^{-1}\bm{\xi}(t) $, which results in
    $$ \nu(t)=\underbrace{\textnormal{\col}(e^{\lambda_1 t}\nu_1(0),\dots, e^{\lambda_n t}\nu_n(0))}_{e^{\Lambda t}\nu(0)}.$$
    Hence, the time-varying matrix $\Pi(t)$ admits the form
    \begin{align*}
    \Pi(t)=&\ \left[\begin{matrix}\bm{\xi}(t)  &\Phi(a) \bm{\xi}(t) &\cdots & \Phi(a)^{n-1} \bm{\xi}(t) \end{matrix}\right]\\
    %=&\ \left[\begin{matrix}P_{\Lambda}^{-1}\bm{\xi}(t)  &P_{\Lambda}^{-1}\Phi(a) P_{\Lambda}^{-1}\bm{\xi}(t) &\dots & \Phi(a)^{n -1} P_{\Lambda}^{-1}\bm{\xi}(t) \end{matrix}\right]\\
    =&\ P_{\Lambda} \left[\begin{matrix}\nu(t)  &\Lambda \nu(t) &\cdots & \Lambda^{n -1}\nu(t) \end{matrix}\right]\\
     %=&\ P_{\Lambda}^{-1} \textnormal{diag}(e^{\lambda_1 t}\nu_1(0),\dots, e^{\lambda_n t}\nu_n(0))\underbrace{\left[\begin{matrix}1  & \lambda_1 & \dots &\lambda_1^{n-1}\\
    %1  & \lambda_2 & \dots& \lambda_2^{n-1}\\
   % \vdots &\vdots& \ddots & \vdots \\
   % 1  & \lambda_n & \dots& \lambda_n^{n-1}\\
   % \end{matrix}\right]}_{P_{\Lambda}^T}.
    =&\ P_{\Lambda} \textnormal{diag}(e^{\lambda_1 t}\nu_1(0),\dots, e^{\lambda_n t}\nu_n(0))P_{\Lambda}^{\top}.
    \end{align*}
It is noted from $\bm{\xi}= \col\!\left(\bm{\hat{x}}_2,\frac{d\bm{\hat{x}}_2}{dt},\dots,\frac{d^{n-1}\bm{\hat{x}}_2}{dt^{n-1}}\right)\!$ and \eqref{remPE-trsin} that
\begin{align*}\nu(0)&= P_{\Lambda}^{-1}\bm{\xi}(0)\\
=&\  P_{\Lambda}^{-1} \col\Big(\sum\limits_{i=1}^{n}C_{j}(v(0), w,\sigma),\dots, \sum\limits_{j=1}^{n}C_{j}(v(0), w,\sigma)\lambda_j^{n-1}\Big)\\
=&\  \underbrace{P_{\Lambda}^{-1}P_{\Lambda}}_{I_n}\col(C_{1}(v(0), w,\sigma),\dots, C_{n}(v(0), w,\sigma)).
\end{align*}
From Assumption \ref{ass5-explicit},  for any $v(0)\in \mathds{V}$, $w\in \mathds{W}$, and $\sigma\in \mathds{S}$, $C_{i}(v(0), w,\sigma)\neq 0$ results in $\nu_i(0)\neq 0$, for $i=1,\dots,n$.
Hence, the matrix $\Pi(t)$ is nonsingular due to the fact that $\nu_i(0)\neq 0$, $\forall i in\in\{i,\dots,n\}$. 
Therefore, $\Theta (\theta)$ is nonsingular. Equation \eqref{XIQA-explicit} admits the solution
\begin{align}\label{a-explicit}
a =\underbrace{-\Theta(\theta)^{-1}\textnormal{\col}(\theta_{n+1},\dots,\theta_{2n})}_{\check{a} (\theta)}.
\end{align}
\end{proof}

 From Lemma 3 in \cite{wang2023nonparametric}, the existence of a nonlinear mapping $\chi\left(\eta, \check{a} (\eta )\right)$ strictly relies on the solution of a time-varying equation,  
$$\Theta(\eta)\check{a} (\eta)+\textnormal{\col}(\eta_{n +1},\dots,\eta_{2n})=\textbf{0}.$$
It is noted that $\Theta(\eta(t))$ is not always invertible over $t\geq 0$, and there may be time instants where the inverse of $\Theta(\eta(t))$ may not be well-defined. % as the inverse of $\Theta(\eta(t))$ may not exist for some instants.
%As a result, the inverse of $\Theta(\eta(t))$ may not be well-defined at these instants.

From Assumptions \ref{ass0} and \ref{ass0i}, it follows that $\bm{\eta} ^{\star}$ and $a$ belong to some compact set $\mathds{D}$. For the composite system \eqref{second-nonlinear-systems}, as shown in Fig.~\ref{framework2}, we propose the regulator
\begin{subequations}\label{explicit-mas}
\begin{align}
%\dot{k}_{i} &= \rho_{i}(e_{vi})e_{vi}^{2} \label{Nussam-mas2}\\
\dot{\hat{a}}  &=- k_a \Theta(\eta)^{\!\top}\left[\Theta (\eta)\hat{a} +\textnormal{\col}(\eta_{n +1},\dots{},\eta_{2n})\right],\\
u &= \alpha_{s,r}(\epsilon_1,\epsilon_2,\dots,\epsilon_r,k^*,\eta,\hat{a}),\label{explicit-mas3}
\end{align}\end{subequations}
where $k_a$ is positive constant, $\eta$ is generated in \eqref{explicit-mas1}, $\hat{a}$ is the estimate of the unknown parameter vector $a$, $\rho (\cdot)\geq 1$ is a positive smooth function, $\epsilon_1=e$,%, and $\epsilon_1=e$ and $\epsilon_{i+1}=\hat{x}_{i+1}-\alpha_{s,i}$ with 
\begin{align}\label{alpha-function-s}
%\quad\quad\quad\quad\quad \quad
\epsilon_{i+1}=&\ \hat{x}_{i+1}-\alpha_{s,i}(\epsilon_1,\dots,\epsilon_i,k^*,\eta,\hat{a},\hat{x}_1),\nonumber\\
\alpha_{s,1}( \epsilon_1,k^*,\eta,\hat{a})=&-k^*\rho(\epsilon_1)\epsilon_1+\chi_s(\eta,\hat{a}),\nonumber\\
%\Delta_1( \epsilon_1,\bar{Z},\mu)=&b^2\bar{x}_2^2 +\bar{g}_1^2(\bar{z}, \epsilon_1, \mu)+b^2\bar{\chi}^2(\bar{\eta},\bar{a},e,\mu)\nonumber\\
\alpha_{s,2}(\epsilon_1,\epsilon_2,k^*,\eta,\hat{a},\hat{x}_1)%=&-b\epsilon_1-\epsilon_2+\lambda_{2}\hat{x}_1+\frac{\partial \alpha_1}{\partial \epsilon_1}\dot{e}_1+\frac{\partial \alpha_1 }{\partial \eta}\dot{\eta}\nonumber\\
%&+\frac{\partial \alpha_1}{\partial e}\Big[b(\hat{x}_2-\bm{\hat{x}}_2)+b\bar{x}_2 +\bar{g}_1(\bar{z}, e, \mu)\Big]\nonumber\\
=&-b\epsilon_1-\epsilon_2+\lambda_{2}\hat{x}_1+\frac{\partial \alpha_{s,1} }{\partial \eta}\dot{\eta}\nonumber\\
&+b\frac{\partial \alpha_{s,1}}{\partial \epsilon_1}(\epsilon_2-{k}^*\rho(\epsilon_1)\epsilon_1)\nonumber\\
&-\frac{1}{2}\epsilon_2\!\left(\frac{\partial \alpha_{s,1}}{\partial \epsilon_1}\right)^{\!2}+\frac{\partial \alpha_{s,1} }{\partial \hat{a}}\dot{\hat{a}}\nonumber\\
&+\frac{\partial \alpha_{s,1} }{\partial k^*}\dot{k}^*,\nonumber\\
%\alpha_3(\epsilon_1,\epsilon_2,\epsilon_{3},k^*,\eta,\hat{x}_1)=&-e_{2}-\epsilon_{3}+\lambda_{3}\hat{x}_{1}+\frac{\partial \alpha_{2} }{\partial \eta}\dot{\eta}\nonumber\\
%&+\frac{\partial \alpha_2 }{\partial \epsilon_2}\dot{e}_2+\frac{\partial \alpha_2 }{\partial \hat{x}_1}\dot{\hat{x}}_1\nonumber\\
%&+b\frac{\partial \alpha_2}{\partial \epsilon_1}(\hat{x}_2-\chi(\eta))-\frac{1}{2}\epsilon_{3}\left(\frac{\partial \alpha_2}{\partial \epsilon_1}\right)^2\nonumber\\
\alpha_{s,i}(\epsilon_1,\dots,\epsilon_i,k^*,\eta,\hat{a},\hat{x}_1)=&-\epsilon_{i-1}-\epsilon_i+\lambda_{i}\hat{x}_{1}+\frac{\partial \alpha_{s,i-1} }{\partial \eta}\dot{\eta}\nonumber\\
&+\frac{\partial \alpha_{s,i-1} }{\partial \hat{x}_1}\dot{\hat{x}}_1+ \sum_{j=2}^{i-1}\frac{\partial \alpha_{s,i-1} }{\partial \epsilon_{j}}\dot{\epsilon}_j\nonumber\\
 &+b\frac{\partial \alpha_{s,i-1}}{\partial \epsilon_1}(\epsilon_2- {k}^*\rho(\epsilon_1)\epsilon_1)\nonumber\\
 &-\frac{1}{2}\epsilon_i\!\left(\frac{\partial \alpha_{s,i-1}}{\partial \epsilon_1}\right)^{\!2}\nonumber\\
 &+\frac{\partial \alpha_{s,i} }{\partial \hat{a}}\dot{\hat{a}}+\frac{\partial \alpha_{s,1} }{\partial k^*}\dot{k}^*,\nonumber\\
& \quad\quad\quad\quad\quad i=3,\dots,r,
\end{align}
where $\dot{k}^*$ will be zero when $k^*$ is a constant, $\hat{a}$ is generated in \eqref{explicit-mas}, $\hat{x}_{1}, \dots, \hat{x}_{r}$ and $\eta$ are generated in \eqref{input-filter} and \eqref{explicit-mas1}, respectively.
The smooth function $\chi_s(\eta,\hat{a})$ is given by
%\begin{align}
%\chi_{i,s}(\eta_i, \hat{a}_i)=\left\{ \begin{array}{cc}
 %                                  \chi_i(\eta_i, \hat{a}_i),& \textnormal{if}~(\eta_i,\hat{a}_i)\in \mathds{D};\\
 %                                  0,& \textnormal{if}~(\eta,\hat{a})\notin \mathds{B}_i;\\
 %                                \end{array}
%\right.
%\end{align}
%with $\mathds{B}_i=\{(\eta_i,\hat{a}_i)| \|\col(\eta_i, \hat{a}_i)\|^2\leq \delta_i+1\}$, $\delta_i=\max\limits_{(\eta_i,\hat{a}_i)\in \mathds{D}}\|\col(\eta_i, \hat{a}_i)\|^2$ and
%
%A specific design can be obtained by choosing the following:
\begin{align}\label{chisatu}
\chi_{s}(\eta, \hat{a})=\chi(\eta, \hat{a})\Psi(\delta+1-\|\col(\eta, \hat{a})\|^2),
\end{align}
where
$$\chi(\eta,\hat{a})\equiv\Gamma \Xi (\hat{a})\textnormal{\col}(\eta_{1},\dots{},\eta_{n}),$$
with $\Psi(\varsigma)=\frac{\psi(\varsigma)}{\psi(\varsigma)+\psi(1-\varsigma)}$, $\delta=\max\limits_{(\eta,\hat{a})\in \mathds{D}}\|\col(\eta, \hat{a})\|^2$ and
$$\psi(\varsigma)=\left\{\begin{array}{cc}e^{-1/\varsigma} & \textnormal{for}\;\varsigma >0,\\
0 & \textnormal{for}\;\varsigma \leq0.
\end{array}\right.$$
Performing the coordinate/input transformations
\begin{align*}
&\bar{\eta}_e =\bar{\eta}+b^{-1}N\epsilon_1,\quad\bar{a}=\hat{a}-a,\quad\hat{x}_2=\epsilon_2+\alpha_{s,1},\\
&\bar{\chi}_s(\bar{\eta}_e,\bar{a},\mu)= \chi_s(\bar{\eta}_e+\bm{\eta}^*,\bar{a}+a)-\chi(\bm{\eta}^*,a),
\end{align*}
leads to the augmented system: 
\begin{subequations}\label{bara-deriv}\begin{align}
\dot{\bar{z}} 
=&\ \bar{f}(\bar{z},\epsilon_1,\mu),\\
\dot{\bar{x}}_c= &\ M_c\bar{x}_c+\bar{G}_c(\bar{z}, \epsilon_1, \mu), \\
\dot{e}_1=&\ b(\epsilon_2-k^*\rho(\epsilon_1)\epsilon_1)+b\bar{\chi}_s(\bar{\eta}_e,\bar{a},\mu)\nonumber\\
&\ +b\bar{x}_2 +\bar{g}_1(\bar{z}, \epsilon_1, \mu), \\
 \dot{\hat{x}}_i=&\ \hat{x}_{i+1}-\lambda_i\hat{x}_1,\;\;i=2,\dots,r-1,\;\;\\
  \dot{\hat{x}}_r=&\ u-\lambda_r\hat{x}_1,\\
\dot{\bar{a}}=&-k_a \Theta(\bm{\eta}^{\star})^{\!\top}\Theta (\bm{\eta}^{\star} )\bar{a} - k_1 \bar{O}(\bar{\eta}_e,\bar{a}),\label{bara-deriv-a}
\end{align} \end{subequations}
where 
\begin{align*}
\bar{O}(\bar{\eta}_e,\bar{a})=&\
\Theta(\bm{\eta}^{\star})^{\!\top}\Theta (\bar{\eta}_e)\bar{a}+\Theta (\bar{\eta}_e)^{\!\top}\Theta (\bm{\eta}^{\star} )\bar{a}\\
&+\Theta(\bar{\eta}_e)^{\!\top} 
\Theta(\bar{\eta}_e)a +\Theta(\bar{\eta}_e)^{\!\top}
\Theta(\bar{\eta}_e)\bar{a}\\
&+\Theta(\bm{\eta}^{\star})^{\!\top}
\Theta(\bar{\eta}_e){a} +
\Theta(\bm{\eta}^{\star})^{\!\top} \textnormal{\col}(\bar{\eta}_{e,n +1},\dots{},\bar{\eta}_{e,2n})\\
&+\Theta(\bar{\eta}_e)^{\!\top}\textnormal{\col}(\bar{\eta}_{e,n +1},\dots{},\bar{\eta}_{e,2n}).
\end{align*}
Moreover, the $\bar{a}$-subsystem \eqref{bara-deriv-a} has a form similar to system (27d) in \cite{wang2023nonparametric}. 
As a result,  the $\bar{a}$-subsystem \eqref{bara-deriv-a}, under Assumptions \ref{ass0}, \ref{ass1}, and \ref{H2}, admits the following lemma (see Lemma 4 in \cite{wang2023nonparametric}).
\begin{lem}\label{lemmabodev}
For the system \eqref{bara-deriv-a} under Assumptions \ref{ass0}, \ref{ass1}, \ref{ass0i}, and \ref{ass5-explicit}, Properties 3 and 4 are satisfied:
\begin{prop}\label{property3}%
There are smooth integral Input-to-State Stable Lyapunov functions $V_{\bar{a}}\equiv  V_{\bar{a}}\big(\bar{a}\big)$ satisfying
\begin{align}
\underline{\alpha}_{\bar{a}}(\|\bar{a}\|^2)&\leq V_{\bar{a}}(\bar{a})\leq \bar{\alpha}_{\bar{a}}(\|\bar{a}\|^2),\notag\\
\dot{V}_{\bar{a}}\big|_{\eqref{explicit-mas1}} &\leq  -\alpha_{\bar{a}}(V_{\bar{a}}) +c_{ae}\|\bar{Z}\|^2+c_{ae}e^2,\label{V2-tildea}
\end{align}
for positive constant $c_{ae}$, and comparison functions $\underline{\alpha}_{\bar{a}}(\cdot)\in \mathcal{K}_{\infty}$, $\bar{\alpha}_{\bar{a}}(\cdot)\in \mathcal{K}_{\infty}$, $\alpha_{\bar{a}}(\cdot)\in \mathcal{K}_{o}$.
\end{prop}
\begin{prop} \label{property4} There are positive constants $\phi_0$, $\phi_1$, and $\phi_2$ such that
\begin{align*}
|b\bar{\chi}_s(\bar{\eta}_e,\bar{a})|^2\leq &\ \phi_0e^2+\phi_1\|\bar{Z}\|^2+\phi_{2}\alpha_{\bar{a}}(V_{\bar{a}}).\end{align*}
\end{prop}
\end{lem}
\begin{rem}\label{rem:proof_lemma_bodev}
 Under Assumption \ref{ass5-explicit}, the matrix $\Theta(\bm{\eta}^\star)$ is bounded and satisfies $\Theta(\bm{\eta}^\star)^\top \Theta(\bm{\eta}^\star) \geq \Theta_m I > 0$.
The perturbation term $\bar{O}(\bar{\eta}_e, \bar{a})$ in \eqref{bara-deriv} satisfies the bound $\|\bar{O}\| \leq c_1 \|\bar{a}\| (\|\bar{a}\| + 1) \|\bar{\eta}_e\|$ for some $c_1 > 0$.
To counteract the growth of $\bar{a}$ in the perturbation, we employ the Lyapunov function $V_{\bar{a}}(\bar{a}) = \ln(1+\|\bar{a}\|^2)$.
Its time derivative along \eqref{bara-deriv} satisfies
\begin{align*}
    \dot{V}_{\bar{a}} &= \frac{-2k_1 \bar{a}^\top \Theta^\top \Theta \bar{a} - 2k_1 \bar{a}^\top \bar{O}}{1+\|\bar{a}\|^2} \\
    &\leq \frac{-k_1 \Theta_m \|\bar{a}\|^2 + 2k_1 c_1 \|\bar{a}\|^2 (\|\bar{a}\|+1) \|\bar{\eta}_e\|}{1+\|\bar{a}\|^2} \\
    &\leq -\underbrace{\frac{k_1 \Theta_m \|\bar{a}\|^2}{2(1+\|\bar{a}\|^2)}}_{\alpha_{\bar{a}}(V_{\bar{a}})} + c_{ae}\|\bar{\eta}_e\|^2,
\end{align*}
where $c_{ae}$ is a sufficiently large constant derived using Young's inequality. Since $\|\bar{\eta}_e\|^2 \leq \ell_1 \|\bar{Z}\|^2 + \ell_2 e^2$, inequality \eqref{V2-tildea} follows.

For Property \ref{property4}, note that $\bar{\chi}_s(\bar{\eta}_e, \bar{a})$ is smooth and vanishes at the origin $(\bar{\eta}_e, \bar{a})=(0,0)$.
Also, $b$ is constant. Then, there exist class $\mathcal{K}_\infty$ functions $\gamma_a(\cdot)$ and $\gamma_\eta(\cdot)$ such that 
$$ |b\bar{\chi}_s(\bar{\eta}_e, \bar{a})|^2 \leq \gamma_\eta(\|\bar{\eta}_e\|^2) + \gamma_a(\|\bar{a}\|^2). $$
Using the property that $\alpha_{\bar{a}}(V_{\bar{a}})$ behaves linearly for small $\bar{a}$ and saturates for large $\bar{a}$, and considering the compact support or boundedness of $\chi_s$ in the proposed design, the term $\gamma_a(\|\bar{a}\|^2)$ can be dominated by $\phi_2 \alpha_{\bar{a}}(V_{\bar{a}})$. Substituting $\|\bar{\eta}_e\|^2$ with linear combinations of $\|\bar{Z}\|^2$ and $e^2$ yields the result.
\end{rem}
\begin{thm}\label{Thm-Lemma-nonpara}% 
For system \eqref{Main-sys-com},
under Assumptions \ref{ass0}--\ref{ass5-explicit}, there is a sufficiently large positive smooth function $\rho(\cdot)$ and a positive real number $k^*$ such that the controller
\begin{align}
u &= -\alpha_{s,r}(\epsilon_1,\epsilon_2,\dots,\epsilon_r,k^*,\eta,\hat{a}) ,\label{nonpara-control} 
%\dot{\hat{k}}&=k(\zeta)\zeta^2
\end{align}
 solves Problem \ref{Prob: second-order-Output-regulation}, and there exists a continuous
positive definite function $U\equiv U(\bar{Z}, \epsilon_1, \dots,\epsilon_r, \bar{a})$
such that, for all $\mu\in \mathds{S}\times \mathds{V}\times \mathds{W}$,
\begin{align}\label{dotU}\dot{U}\leq -\big\|\bar{Z}\big\|^2-\sum\nolimits_{j=1}^{r}\epsilon_{j}^2-\alpha_{\bar{a}}(V_{\bar{a}}).\end{align}
%for any $k_0\geq k^*$. 
\end{thm}
\begin{proof} From Property \ref{property1}, the changing supply rate technique \citep{sontag1995changing} can be applied to show that, given any smooth function $\Delta_{Z}(\bar{Z})>0$, there exists a continuous function $V_{1}(\bar{Z})$ satisfying
$$\underline{\alpha}_{1}\big(\big\|\bar{Z} \big\|^2\big)\leq V_{1}\big( \bar{Z}  \big)\leq\overline{\alpha}_{1}\big(\big\|\bar{Z} \big\|^2\big)$$
 for some class $\mathcal{K}_{\infty}$ functions $\underline{\alpha}_{1}(\cdot)$ and $\overline{\alpha}_{1}(\cdot)$ such that, for all $\mu\in \Sigma$, along the trajectories of the $Z $ subsystem, $$\dot{V}_{1} \leq-\Delta_{Z}(\bar{Z} )\big\|\bar{Z} \big\|^2+ \hat{\gamma}^* \hat{\gamma} \left(\epsilon_1\right)\epsilon_1^2, $$
where $\hat{\gamma}^*$ is known positive constant and $\hat{\gamma} \left(\cdot\right)\geq 1$ is a known smooth positive definite function.

Define the Lyapunov function $U_1(\bar{Z}, \epsilon_1)=V_{1}\big( \bar{Z}  \big)+ \epsilon_1^2$. Then, the time derivative of $U_1\equiv U_1(\bar{Z}, \epsilon_1)$ along the trajectory of $\epsilon_1$-subsystem with $\hat{x}_2=\epsilon_2+\alpha_1$ and $\eta=\bar{\eta}+\bm{\eta}^{\star}+Nb^{-1}\epsilon_1 $ leads to
\begin{align}\label{U1-derivative}
\dot{U}_1(\bar{Z}, \epsilon_1)=&\ \dot{V}_{1}\big( \bar{Z}  \big)+ 2\epsilon_1\dot{\epsilon}_1\nonumber\\
\leq &-\Delta_{Z}(\bar{Z} )\big\|\bar{Z} \big\|^2+ \hat{\gamma}^* \hat{\gamma} \left(\epsilon_1\right)\epsilon_1^2+2\epsilon_1\bar{g}_1(\bar{z}, \epsilon_1, \mu) \nonumber\\
&+2b\epsilon_1(\underbrace{\epsilon_2+\alpha_1( \epsilon_1,k^*,\eta)}_{\hat{x}_2}-\chi(\bm{\eta}^*))+2b\epsilon_1\bar{x}_2\nonumber\\
\leq &-\Delta_{Z}(\bar{Z} )\big\|\bar{Z} \big\|^2+ \hat{\gamma}^* \hat{\gamma} \left(\epsilon_1\right)\epsilon_1^2+2\epsilon_1\bar{g}_1(\bar{z}, \epsilon_1, \mu)\nonumber\\
&+2b\epsilon_1(\epsilon_2+\underbrace{\alpha_1( \epsilon_1,k^*,\eta)}_{-k^*\rho(\epsilon_1)\epsilon_1+\chi(\eta)}-\chi(\bm{\eta}^*))+2b\epsilon_1\bar{x}_2 \nonumber\\
\leq &-\Delta_{Z}(\bar{Z} )\big\|\bar{Z} \big\|^2-\big(2b k^*\rho(\epsilon_1)-\hat{\gamma}^* \hat{\gamma} \left(\epsilon_1\right)\big)\epsilon_1^2\nonumber\\
&+2b\epsilon_1(\epsilon_2-\bar{\chi}(\bar{\eta},\epsilon_1,\mu))\nonumber\\
&+2b\epsilon_1\bar{x}_2 +2\epsilon_1\bar{g}_1(\bar{z}, \epsilon_1, \mu)\nonumber\\
\leq &-\Delta_{Z}(\bar{Z} )\big\|\bar{Z} \big\|^2-\big(2b k^*\rho(\epsilon_1)-\hat{\gamma}^* \hat{\gamma} \left(\epsilon_1\right)\big)\epsilon_1^2\nonumber\\
&+2b\epsilon_1\epsilon_2-2b\epsilon_1\bar{\chi}(\bar{\eta},\epsilon_1,\mu)\nonumber\\
&+2b\epsilon_1\bar{x}_2 +2\epsilon_1\bar{g}_1(\bar{z}, \epsilon_1, \mu)\nonumber\\
\leq &-\Delta_{Z}(\bar{Z} )\big\|\bar{Z} \big\|^2-\big(2b k^*\rho(\epsilon_1)-3-\hat{\gamma}^* \hat{\gamma} \left(\epsilon_1\right)\big)\epsilon_1^2\nonumber\\
&+2b\epsilon_1\epsilon_2+\Delta_1( \epsilon_1,\bar{Z},\mu)
\end{align}
where 
\begin{align*}
    \Delta_1( \epsilon_1,\bar{Z},\mu)&=b^2\bar{x}_2^2 +\bar{g}_1(\bar{z}, \epsilon_1, \mu)^2+b^2\bar{\chi}(\bar{\eta}, \epsilon_1, \mu)^2,\\
    \bar{\chi}(\bar{\eta}, \epsilon_1, \mu) &\equiv \chi(\bar{\eta}+\bm{\eta}^*+Nb^{-1} \epsilon_1)-\chi(\bm{\eta}^*).
\end{align*}
Now let $ U_2(\bar{Z}, \epsilon_1, \epsilon_2)=U_1(\bar{Z}, \epsilon_1)+\epsilon_2^2$. The time derivative of $U_2\equiv U_2(\bar{Z}, \epsilon_1, \epsilon_2)$ along the trajectory of $\epsilon_2$-subsystem with $\hat{x}_3=\epsilon_{3}+\alpha_2$ is given by 
\begin{align}%\label{U1-derivative-1}
\dot{U}_2 \leq &\ \dot{U}_1+2 \epsilon_2\dot{\epsilon}_2\nonumber\\
\leq&-\Delta_{Z}(\bar{Z} )\big\|\bar{Z} \big\|^2-\big(b k^*\rho(\epsilon_1)-3-\hat{\gamma}^* \hat{\gamma} \left(\epsilon_1\right)\big)\epsilon_1^2\nonumber\\
&+2b\epsilon_1\epsilon_2+\Delta_1( \epsilon_1,\bar{Z},\eta)+2 \epsilon_2(\epsilon_{3}+\alpha_2-\lambda_2\hat{x}_1-\dot{\alpha}_1)\nonumber\\
\leq&-\Delta_{Z}(\bar{Z} )\big\|\bar{Z} \big\|^2-\big(b k^*\rho(\epsilon_1)-3-\hat{\gamma}^* \hat{\gamma} \left(\epsilon_1\right)\big)\epsilon_1^2\nonumber\\
&+2b\epsilon_1\epsilon_2+\Delta_1( \epsilon_1,\bar{Z},\eta)\nonumber\\
&+2 \epsilon_2\Big(\epsilon_{3}+\alpha_2-\lambda_2\hat{x}_1-\frac{\partial \alpha_1}{\partial \epsilon_1}\dot{\epsilon}_1-\frac{\partial \alpha_1 }{\partial \eta}\dot{\eta}-\frac{\partial \alpha_1 }{\partial k^*}\dot{k}^*\Big)\nonumber\\
\leq&-\Delta_{Z}(\bar{Z} )\big\|\bar{Z} \big\|^2-\big(b k^*\rho(\epsilon_1)-3-\hat{\gamma}^* \hat{\gamma} \left(\epsilon_1\right)\big)\epsilon_1^2\nonumber\\
&+2b\epsilon_1\epsilon_2+\Delta_1( \epsilon_1,\bar{Z},\eta)+2 \epsilon_2\epsilon_{3}\nonumber\\
&+2 \epsilon_2\Big(\alpha_2-\lambda_2\hat{x}_1-\frac{\partial \alpha_1 }{\partial \eta}\dot{\eta}-b\frac{\partial \alpha_1}{\partial \epsilon_1}(\underbrace{\epsilon_2+\alpha_1}_{\hat{x}_2}-\chi(\eta))\nonumber\\
&-\frac{\partial \alpha_1}{\partial \epsilon_1}\big[b\bar{\chi}(\bar{\eta},\epsilon_1,\mu)+b\bar{x}_2 +\bar{g}_1(\bar{z}, \epsilon_1, \mu)\big]-\frac{\partial \alpha_1 }{\partial k^*}\dot{k}^*\Big)\nonumber\\
\leq&-\Delta_{Z}(\bar{Z} )\big\|\bar{Z} \big\|^2-\big(b k^*\rho(\epsilon_1)-3-\hat{\gamma}^* \hat{\gamma} \left(\epsilon_1\right)\big)\epsilon_1^2\nonumber\\
&+\Delta_1( \epsilon_1,\bar{Z},\eta)+2 \epsilon_2\epsilon_{3}\nonumber\\
&+2 \epsilon_2\Big(\alpha_2-\epsilon_2+ \underbracea{\epsilon_2+b\epsilon_1-\lambda_2\hat{x}_1-\frac{\partial \alpha_1 }{\partial \eta}\dot{\eta}}\nonumber\\
&\underbraceb{{}-b\frac{\partial \alpha_1}{\partial \epsilon_1}(\epsilon_2\underbrace{-k^*\rho(\epsilon_1)\epsilon_1+\chi(\eta)}_{\alpha_1}-\chi(\eta))-\frac{\partial \alpha_1 }{\partial k^*}\dot{k}^*{}}_{-\alpha_2 }\nonumber\\
&\underbraced{{}+\frac{1}{2}\epsilon_2\Big(\frac{\partial \alpha_1}{\partial \epsilon_1}\Big)^2}\Big)+\epsilon^2_2\nonumber\\
&+\underbrace{\big[b^2\bar{\chi}(\bar{\eta},\epsilon_1,\mu)^2+b^2\bar{x}_2^2 +\bar{g}_1(\bar{z}, \epsilon_1, \mu)^2\big]}_{\Delta_1( \epsilon_1,\bar{Z},\eta)}\nonumber\\
\leq&-\Delta_{Z}(\bar{Z} )\big\|\bar{Z} \big\|^2+2\Delta_1( \epsilon_1,\bar{Z},\mu)+2 \epsilon_2\epsilon_{3}\nonumber\\
&-\big(2b k^*\rho(\epsilon_1)-3-\hat{\gamma}^* \hat{\gamma} \left(\epsilon_1\right)\big)\epsilon_1^2-\epsilon_2^2.\nonumber
\end{align}
Now let $U_i(\bar{Z}, \epsilon_1, \dots,\epsilon_i)=U_{i-1}(\bar{Z}, \epsilon_1,\dots,\epsilon_{i-1})+\epsilon_i^2$.  The time derivative of $U_i\equiv U_i(\bar{Z}, \epsilon_1,\dots,e_{i})$ along the trajectory of $\epsilon_i$-subsystem with $\hat{x}_{i+1}=\epsilon_{i+1}+\alpha_i$ is given by 
\begin{align}%\label{U1-derivative-i}
\dot{U}_i\leq&-\Delta_{Z}(\bar{Z} )\big\|\bar{Z} \big\|^2+i\Delta_1( \epsilon_1,\bar{Z},\mu)+2 \epsilon_i\epsilon_{i+1}\nonumber\\
&-\big(2b k^*\rho(\epsilon_1)-3-\hat{\gamma}^* \hat{\gamma} \left(\epsilon_1\right)\big)\epsilon_1^2-\sum\nolimits_{j=2}^{i}\epsilon_{j}^2.\nonumber
\end{align}
Finally, at $i=r$ results in
\begin{align}\label{U1-derivative-r}
U_r(\bar{Z}, \epsilon_1, \dots,\epsilon_r)=V_{1}\big( \bar{Z}  \big) +\sum\nolimits_{i=1}^{r}\epsilon_i^2.
\end{align}
The time derivative of $U_r\equiv U_r(\bar{Z}, \epsilon_1, \dots,\epsilon_r)$ along the trajectory of systems \eqref{alpha-function-s} and \eqref{bara-deriv} with ${e}_{r+1}=0$ is given by 
\begin{align}\label{U1-derivative-sr}
\dot{U}_r\leq&-\Delta_{Z}(\bar{Z} )\big\|\bar{Z} \big\|^2+r\Delta_{a}( \epsilon_1,\bar{Z},\bar{a},\mu)\nonumber\\
&-\big(2b k^*\rho(\epsilon_1)-3-\hat{\gamma}^* \hat{\gamma} \left(\epsilon_1\right)\big)\epsilon_1^2-\sum\nolimits_{j=2}^{r}\epsilon_{j}^2,
\end{align} 
where $$\Delta_a( \epsilon_1,\bar{Z},\bar{a},\mu)=b^2\bar{x}_2^2 +\bar{g}_1^2(\bar{z}, \epsilon_1, \mu)+b^2\bar{\chi}_s^2(\bar{\eta}_e,\bar{a},\mu).$$
From Properties \ref{property2} and \ref{property4}, there are positive smooth functions $\gamma_{g0}(\cdot)$ and $\gamma_{g1}(\cdot)$, positive constants $\phi_0$, $\phi_1$, and $\phi_2$ such that
\begin{align*}\Delta_a( \epsilon_1,\bar{Z},\bar{a},\mu)\leq & \ (\gamma_{g0}(\bar{Z})+\phi_1)\|\bar{Z}\|^2\nonumber\\
&+\epsilon_1^2(\gamma_{g1}(\epsilon_1)+\phi_0)+\phi_{2}\alpha_{\bar{a}}(V_{\bar{a}}). 
\end{align*} 
Now, define a Lyapunov function by
 \begin{align}\label{U-zer-a}
 U(\bar{Z}, \epsilon_1, \dots,\epsilon_r, \bar{a})=U_r(\bar{Z}, \epsilon_1, \dots,\epsilon_r)+\phi_{\bar{a}} V_{\bar{a}}(\bar{a}),
\end{align}
 where the positive constant $\phi_{\bar{a}}$ is to be specified, and the integral Input-to-State Stable Lyapunov functions $V_{\bar{a}}(\bar{a})$ is given in Property \ref{property3} of Lemma \ref{lemmabodev}. 
The time derivative of $U\equiv U(\bar{Z}, \epsilon_1, \dots,\epsilon_r, \bar{a})$ along the trajectories of \eqref{bara-deriv} with the control input \eqref{nonpara-control} satisfies
\begin{align}\label{Uderivative}
\dot{U}=&\ \dot{U}_r+\phi_{\bar{a}} \dot{V}_{\bar{a}}(\bar{a})\nonumber\\
\leq & -\big(\Delta_{Z}(\bar{Z} )-r\gamma_{g0}(\bar{Z})-r\phi_1-\phi_{\bar{a}}c_{ae}\big)\big\|\bar{Z} \big\|^2\nonumber\\
&-\big(2b k^*\rho(\epsilon_1)-3-\hat{\gamma}^* \hat{\gamma} \left(\epsilon_1\right)-r\gamma_{g1}(\epsilon_1)-r\phi_0\nonumber\\
&-\phi_{\bar{a}}c_{ae}\big)\epsilon_1^2-\sum\nolimits_{j=2}^{r}\epsilon_{j}^2-(\phi_{\bar{a}}-r\phi_{2})\alpha_{\bar{a}}(V_{\bar{a}}).
   \end{align}
Let the parameter and the smooth functions be
\begin{align} \label{rho-inquality-2}
\phi_{\bar{a}}& \geq r\phi_{2}+1,\notag\\
\Delta_{Z}(\bar{Z} ) &\geq r\gamma_{g0}(\bar{Z})+r\phi_1+\phi_{\bar{a}}c_{ae}+1,\notag\\
\rho(\epsilon_1)
&\geq \max\{\gamma_{g1}(\epsilon_1), \hat{\gamma} \left(\epsilon_1\right), 1\},\notag\\
 k^* &\geq  {(3+ \hat{\gamma}^*+ r+ r\phi_0+\phi_{\bar{a}}c_{ae}+1)}/({2b}).
\end{align}
Hence, \eqref{Uderivative} yields the inequality
\begin{align}\label{strlyafun}\dot{U}\leq -\big\|\bar{Z}\big\|^2-\sum\nolimits_{j=1}^{r}\epsilon_{j}^2-\alpha_{\bar{a}}(V_{\bar{a}}).\end{align}
Finally, because $U(\bar{Z}, \epsilon_1, \dots,\epsilon_r, \bar{a})$ is positive definite and radially unbounded and satisfies a strict Lyapunov function satisfying inequality \eqref{strlyafun}, it follows that the closed-loop system is uniformly asymptotically stable for all $\col(v,w,\sigma)\in \mathds{V}\times \mathds{W}\times \mathds{S}$.
 This completes the proof.
\end{proof}

From Theorem \ref{Thm-Lemma-nonpara}, we can also use the adaptive method to estimate the gain $k^*$. As a result, Theorem \ref{Thm-Lemma-nonpara} can admit the following corollary.
\begin{Corollary}\label{Theorem-4}%
For system \eqref{Main-sys-com} under Assumptions \ref{ass0}--\ref{ass5-explicit}, there is a sufficiently large enough positive smooth function $\rho(\cdot)$ that the controller, 
\begin{subequations}\label{adnon}\begin{align}
u &= \alpha_{s,r}(\epsilon_1,\epsilon_2,\dots,\epsilon_r,\hat{k},\eta,\hat{a}), \label{adnon-1b}\\
\dot{\hat{k}}&=\rho(\epsilon_1)\epsilon_1^2,
\end{align}\end{subequations}
solves Problem \ref{Prob: second-order-Output-regulation} with the functions $\alpha_{s,1}( \epsilon_1,\hat{k},\eta,\hat{a})$, $\alpha_{s,2}(\epsilon_1,\epsilon_2,\hat{k},\eta,\hat{x}_1,\hat{a})$,  and $\alpha_{s,i}(\epsilon_1,\dots,\epsilon_i,\hat{k},\eta,\hat{x}_1,\hat{a})$ defined in \eqref{alpha-function-s}, for $i=3,\dots,r$.
\end{Corollary}%
\begin{proof} Define the Lyapunov function
\begin{align}
V_{r}(\bar{Z}, \epsilon_1,\dots, \epsilon_r, \bar{a},\tilde{k})=\underbrace{\overbrace{V_{1}\big( \bar{Z}  \big) +\sum\nolimits_{i=1}^{r}\epsilon_i^2}^{U_r(\bar{Z}, \epsilon_1, \dots,\epsilon_r)}+\phi_{\bar{a}} V_{\bar{a}}(\bar{a})}_{U(\bar{Z}, \epsilon_1, \dots,\epsilon_r,\bar{a})}+b\tilde{k}^2
\end{align}
for the system \eqref{bara-deriv}, where $\tilde{k}=\hat{k}-k^*$ and $U_r(\bar{Z}, \epsilon_1, \dots,\epsilon_r)$ are the same as in \eqref{U1-derivative-r} and $U(\bar{Z}, \epsilon_1, \dots,\epsilon_r,\bar{a})$ is the same as that in \eqref{U-zer-a}. It can be verified that $V_{r}(\bar{Z}, \epsilon_1,\dots, \epsilon_r, \bar{a},\tilde{k})$ is globally positive definite and radially unbounded. Then, the time derivative of $V_{r}(\bar{Z}, \epsilon_1,\dots, \epsilon_r, \bar{a}, \tilde{k})$ along the trajectories of \eqref{bara-deriv} with the control input \eqref{adnon} satisfies
\begin{align}\label{Vrderivative}
\dot{V}_{r}=&\ \underbrace{\dot{U}_r+\phi_{\bar{a}} \dot{V}_{\bar{a}}(\bar{a})}_{\dot{U}} +2b\tilde{k}\dot{\tilde{k}}\nonumber\\
%\dot{U}=&\ \nonumber\\
\leq & -\big(\Delta_{Z}(\bar{Z} )-r\gamma_{g0}(\bar{Z})-r\phi_1-\phi_{\bar{a}}c_{ae}\big)\big\|\bar{Z} \big\|^2\nonumber\\
&-\big(2b (\underbrace{\hat{k}-k^*}_{\tilde{k}}+k^*)\rho(\epsilon_1)-3-\hat{\gamma}^* \hat{\gamma} \left(\epsilon_1\right)\nonumber\\
&-r\gamma_{g1}(\epsilon_1)-r\phi_0-\phi_{\bar{a}}c_{ae}\big)\epsilon_1^2-\sum\nolimits_{j=2}^{r}\epsilon_{j}^2\nonumber\\
&-(\phi_{\bar{a}}-r\phi_{2})\alpha_{\bar{a}}(V_{\bar{a}})+2b\tilde{k}\rho(\epsilon_1)\epsilon_1^2
   \end{align}
Let the parameter and the smooth functions be
\begin{align} \label{rho-inquality-Vr}
\phi_{\bar{a}}& \geq r\phi_{2}+1,\notag\\
\Delta_{Z}(\bar{Z} ) &\geq r\gamma_{g0}(\bar{Z})+r\phi_1+\phi_{\bar{a}}c_{ae}+1,\notag\\
\rho(\epsilon_1)
&\geq \max\{\gamma_{g1}(\epsilon_1), \hat{\gamma} \left(\epsilon_1\right), 1\},\notag\\
 k^* &\geq  {(3+ \hat{\gamma}^*+ r+ r\phi_0+\phi_{\bar{a}}c_{ae}+1)}/({2b}).
\end{align}
Hence, \eqref{Vrderivative} yields the inequality
\begin{align}\label{V-strlyafun}
    \dot{V}_{r}\leq -\big\|\bar{Z}\big\|^2 - \sum_{j=1}^{r}\epsilon_{j}^2 - \alpha_{\bar{a}}(V_{\bar{a}}).
\end{align}
Since $\alpha_{\bar{a}}(\cdot)$ is a class $\mathcal{K}$ function, the term on the right-hand side is negative semi-definite. By invoking the LaSalle–Yoshizawa Theorem \cite[Lemma 1]{krstic1995nonlinear}, the boundedness of $V_r(t)$ implies that the signals $\bar{Z}$, $\epsilon_1$, $\dots$, $\epsilon_r$, $\bar{a}$, and $\tilde{k}$ are globally uniformly bounded. Furthermore, it follows that
\begin{equation}
    \lim_{t \to \infty} \Bigg( \big\|\bar{Z}(t)\big\|^2 + \sum_{j=1}^{r}\epsilon_{j}^2(t) + \alpha_{\bar{a}}(V_{\bar{a}}(t)) \Bigg) = 0,
\end{equation}
which implies $\lim_{t \to \infty} \|\bar{Z}(t)\| = 0$, $\lim_{t \to \infty} \epsilon_j(t) = 0$ for $j=1,\dots,r$, and $\lim_{t \to \infty} \alpha_{\bar{a}}(V_{\bar{a}}(t)) = 0$.
Therefore, the controller \eqref{adnon} solves Problem \ref{Prob: second-order-Output-regulation}.
\end{proof}
\section{Practical Examples}\label{numerexam}

 \subsection{Regulation of a Virtual Synchronous Generator (VSG)}
\begin{figure}[htbp]
\centering
\begin{tikzpicture}[scale=0.9, transform shape, >=stealth]
    % Inverter Box
    \draw[thick, fill=mitSilverGray!10] (0,0) rectangle (3.5, 2.5);
    \node[align=center] at (1.75, 2.2) {\textbf{Grid-Forming Inverter}};
    
    % Internal "Virtual" Logic
    \draw[dashed, fill=white] (0.5, 0.5) rectangle (3.0, 1.8);
    \node[align=center, font=\scriptsize] at (1.75, 1.5) {\textbf{Virtual Mechanics}};
    
    % Virtual Flywheel Icon
    \draw[thick, fill=mitBlue!20] (1.75, 1.0) circle (0.4);
    \draw[thick, ->] (1.5, 0.8) arc (220:320:0.35);
    \node[font=\tiny] at (1.75, 1.0) {$J, D$};
    
    % Input u
    \draw[<-, thick] (0, 1.0) -- (-2.25, 1.0) node[above,xshift=1cm] {$u$ (Power Ref)};
    
    % Output Filter
    \draw[thick] (3.5, 1.25) -- (4.5, 1.25);
    \draw[decoration={aspect=0.3, segment length=2mm, amplitude=2mm,coil},decorate, thick, mitBlue] (4.5, 1.25) -- (6.0, 1.25);
    \node[above, mitBlue] at (5.25, 1.45) {$L_{filter}$};
    
    % Grid Connection
    \draw[thick] (6.0, 1.25) -- (7.0, 1.25);
    \draw[thick] (7.0, 1.25) -- (7.0, 0.5);
    \draw[thick] (7.0, 0.5) circle (0.3);
    \node at (7.0, 0.5) {$\sim$}; 
    \node[below] at (7.0, 0.1) {Grid $V_g$};
    
    % Measurement
    \draw[->, dashed, thick] (6.0, 1.25) -- (6.0, 0.2) -- (3.0, 0.2) -- (3.0, 0.5);
    \node[below, font=\scriptsize] at (4.5, 0.2) {Feedback $\delta, P_{out}$};

    % Disturbance (Grid Fault)
    \node[draw, fill=mitPink!30, rounded corners, font=\scriptsize] at (6.5, 2.5) (dist) {Grid Faults $d(t)$};
    \draw[->, thick, snake=snake, segment length=3pt] (dist.south) -- (6.5, 1.3);
\end{tikzpicture}

\vspace{-.2cm}

\caption{Schematic of a virtual synchronous generator. The inverter control emulates the dynamics of a rotating machine ($r=2$) to synchronize with the grid.}
\label{fig:vsg}
\end{figure}
\begin{figure}[htbp]
\centering
\epsfig{figure=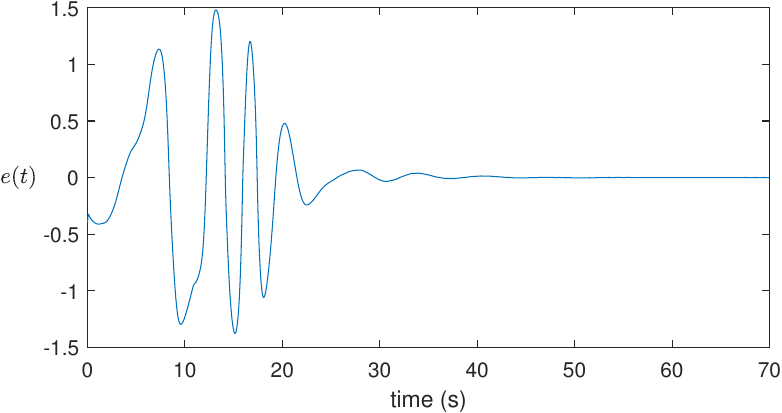,width=0.5\textwidth}
\caption{Tracking performance for the virtual synchronous generator ($r=2$).}
\label{fig:vsg_track}
\quad
\epsfig{figure=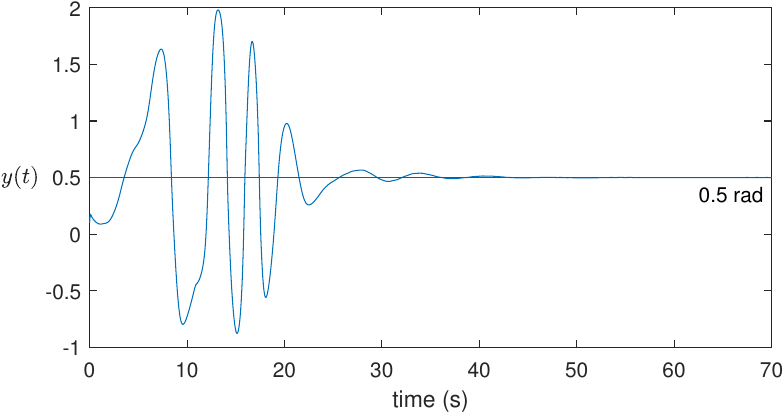,width=0.5\textwidth}
\caption{Desired power angle and trajectory of power angle $x_1(t)$.}
\label{fig:Desired_power_angle}

\quad
\epsfig{figure=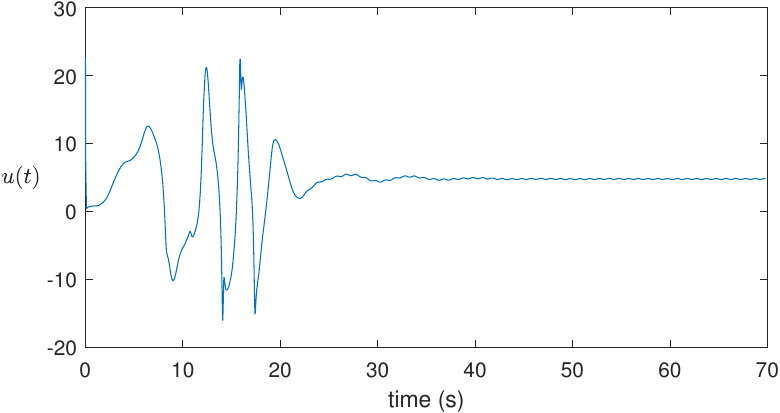,width=0.5\textwidth}
\caption{Control input $u(t)$ for the virtual synchronous generator.}
\label{fig:Control——Input}
\quad
\epsfig{figure=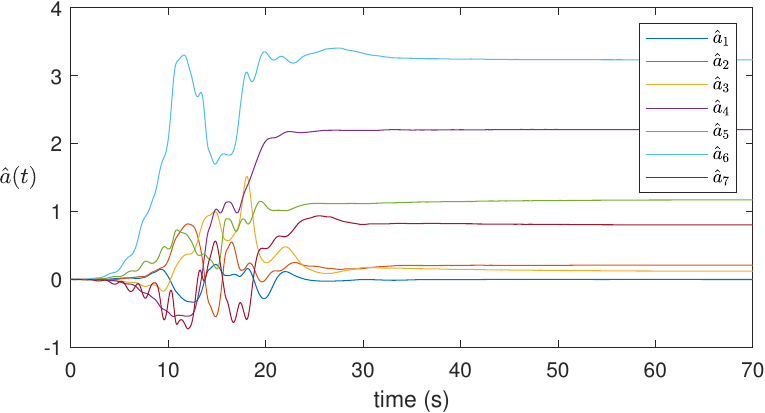,width=0.5\textwidth}
\caption{Parameter estimates for the Virtual Synchronous Generator.}
\label{fig:Parameter_Estimation}
\end{figure}
\begin{figure}[htbp]
\centering
\epsfig{figure=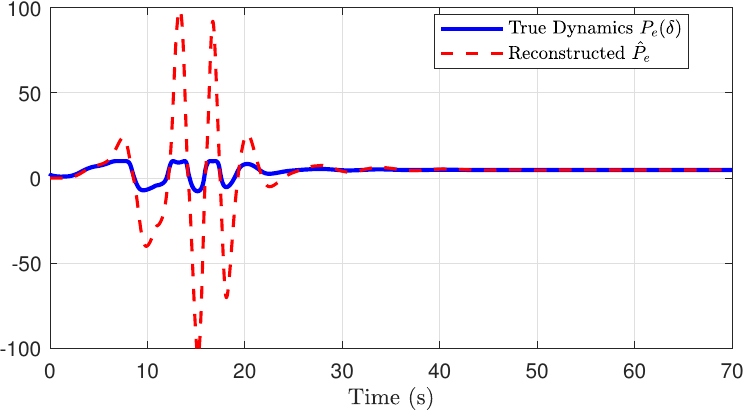,width=0.5\textwidth}
\caption{Visualization of the nonlinearity reconstruction in the time domain.}
\label{fig:Dynamics_Reconstruction_Inverter}
\end{figure}

To demonstrate the application of the proposed framework in modern power systems, consider a grid-connected inverter controlled as a virtual synchronous generator (VSG). The VSG control strategy forces the power electronics interface to emulate the electromechanical dynamics of a traditional synchronous machine, thereby providing inertia and damping support to the grid.

The core dynamics are governed by the virtual swing equation \cite{zhong2010synchronverters}:
\begin{subequations}\label{vsg-dynamics}
\begin{align}
    \dot{x}_1 &= x_2, \\
    \dot{x}_2 &= \frac{1}{J}(u - D x_2 - P_e(x_1) + d(t)), \\
    y &= x_1, \quad e = y - y_r,
\end{align}
\end{subequations}
where $x_1 = \delta$ is the power angle (phase difference between inverter and grid),  $x_2 = \dot{\delta}$ is the frequency deviation from the nominal grid frequency, the control input $u$ represents the mechanical power reference ($P_\textrm{ref}$), and $J$ and $D$ are the virtual inertia and virtual damping coefficients, respectively, which are programmable parameters in the inverter's microcontroller. 

The nonlinearity arises from the electrical power transfer equation $P_e(x_1)$. For a connection to an infinite grid with voltage $V_g$ through a line impedance $X$, this relationship is sinusoidal:
\begin{align*}
    P_e(x_1) = \frac{E_g V_g}{X} \sin(x_1 + \phi),
\end{align*}
where $E_g$ is the inverter output voltage magnitude. The parameters $V_g$, $X$, and $\phi$ are determined by the physical grid conditions and are typically unknown and time-varying (represented by the disturbance $d(t)$).

Differentiating the output $y=x_1$ twice yields
\begin{align*}
    \dot{y} &= x_2, \\
    \ddot{y} &= \frac{1}{J}\left( u - D x_2 - \frac{E_g V_g}{X}\sin(x_1+\phi) + d(t) \right).
\end{align*}
The control input $u$ appears in the second derivative, indicating a relative degree of $r=2$. The control objective is to maintain grid synchronization (regulate $x_1$ to a desired load angle $y_r$ or track a frequency reference) despite fluctuations in grid voltage $V_g$ and impedance $X$. The term $\sin(x_1)$ represents a strong trigonometric nonlinearity that the internal model must learn to ensure stable power delivery.

%In the simulation, the virtual parameters are set to $J=0.1$ and $D=5$. 
In the numerical simulation, the physical and virtual parameters of the VSG system are configured as follows. The virtual inertia and damping coefficient are set to $J=0.1$ and $D=5$, respectively. %These values are selected to provide a fast dynamic response while maintaining sufficient damping to suppress oscillations. 
The electrical parameters characterizing the grid interconnection are chosen as
\begin{itemize}
    \item Inverter output voltage magnitude: $E_g = 10$\,V;
    \item Nominal grid voltage magnitude: $V_g = 10$\,V;
    \item Line impedance: $X = 10\,\Omega$;
    \item Initial phase offset: $\phi = 0$.
\end{itemize}
Based on these parameters, the theoretical maximum power transfer capability is $P_{\max} = \frac{E_g V_g}{X} = 10$. 
While fixed values of $V_g$ and $X$ are used to simulate the plant dynamics, these parameters are treated as \textit{unknown} and potentially time-varying by the proposed nonparametric controller. The controller must learn the unknown sinusoidal relationship $P_e(x_1) = 10 \sin(x_1)$ online to achieve the tracking objective.

To validate the regulation capability under standard dispatch conditions, the control objective is to regulate the power angle to a constant operating point,
\begin{align*}
    y_r(t) = 0.5 \, \text{rad},
\end{align*}
which simulates a fixed power dispatch command from the transmission system operator. Although the reference is constant, the steady-state control input $u_{ss}$ required to maintain this angle is unknown a priori due to the uncertain grid impedance $X$ and disturbances. The internal model must learn the inverse of the nonlinear power transmission characteristic to generate the correct mechanical power reference.
The interaction dynamics between the grid and the VSG, as defined in \eqref{vsg-dynamics}, combined with unknown disturbances, result in an uncertain steady-state behavior. This makes deriving an explicit solution mathematically intractable, particularly since only the system output is available for feedback. To address this, assume that the steady-state input can be described by a linear internal model of the form:
\begin{align*}
\frac{d^{7}\hat{u}}{dt^{7}} + a_7 \frac{d^{6}\hat{u}}{dt^{6}} + \dots + a_2 \frac{d\hat{u}}{dt} + a_1 \hat{u} = 0,
\end{align*}
where $\bm{a} = \col(a_1, a_2, \dots, a_7)$ is a vector of unknown constant coefficients that the controller must identify or adapt to.

To rigorously test the robustness of the proposed nonparametric controller against unknown periodic disturbances, the disturbance $d(t)$ is designed as a composite signal:
\begin{align*}
    d(t) = 0.1 \sin(2\pi t).
\end{align*}
The internal model is expected to learn and compensate for this unknown frequency component ($1\,\text{Hz}$) to achieve asymptotic regulation.

For the control law \eqref{adnon}, we can choose $\rho(e)=1+e^2$ based on \eqref{rho-inquality-2} to make the inequality \eqref{Uderivative} hold %negative definite, 
$k_a=20$ is any positive number in \eqref{explicit-mas}, $\lambda_1=2$ and $\lambda_2=2$ are chosen to make $A$ in \eqref{input-filter} to be Hurwitz. $m_1=1$, $m_2=9.5144$, $m_3=44.7616$,
$m_4=137.7619$, $m_5=309.4184$, $m_6=535.9283$, $m_7=737.6421$, $m_8=819.2345$, $m_9=737.6421$, $m_{10}=535.9283$, $m_{11}=309.4184$, $m_{12}=137.7619$, $m_{13}=44.7616$, and $m_{14}=9.5144$ are chosen to make $M$ in \eqref{explicit-mas1} defined in \eqref{MNINter} to be Hurwitz. The simulation starts with the initial conditions 
$x(0)=\col(0.1,0)$, $\eta(0)=\textbf{0}_{14}$, $\hat{k}(0)=2$, and $\hat{a}(0)=\textbf{0}_7$.

The results in Figs.\  \ref{fig:vsg_track}--\ref{fig:Dynamics_Reconstruction_Inverter} confirm that the VSG successfully synchronizes with the grid, with the internal model reconstructing the unknown grid interaction dynamics $P_e(x_1)$. 
In Fig.\  \ref{fig:Dynamics_Reconstruction_Inverter}, the solid black line represents the true, unknown grid power transfer characteristic $P_e(\delta) = P_{\max}\sin(\delta)$. The colored scatter points represent the instantaneous estimate $\hat{P}_e(t)$ generated by the nonparametric internal model against the actual power angle $\delta(t)$. The convergence of the trajectory (from blue to red) onto the theoretical manifold demonstrates that the internal model has successfully learned the unknown sinusoidal nonlinearity profile without any a priori structural knowledge.

\subsection{Heading Control of a Surface Vessel}

\begin{figure}[htbp]
\centering
\usetikzlibrary{shapes, arrows.meta, decorations.pathmorphing, shadows, calc, positioning}

% \definecolor{mitRed}{RGB}{117,0,20}
% \definecolor{mitSilverGray}{RGB}{139,149,158}
% \definecolor{mitPurple}{RGB}{153,51,255}
% %\definecolor{mitBlue}{RGB}{25,102,255}
% \definecolor{mitGreen}{RGB}{0,173,0}
% \definecolor{mitLGreen}{RGB}{170,255,51}
% \definecolor{mitPink}{RGB}{255,20,240}
% \definecolor{mitLightPink}{RGB}{255,179,255}
% \definecolor{mitDPink}{RGB}{117,0,98}
% \definecolor{mitLightBlue}{RGB}{153,235,255}

\definecolor{oceanBlue}{RGB}{25,102,255}
\definecolor{deepWater}{RGB}{10, 50, 90}
\definecolor{shipDeck}{RGB}{220, 220, 220}
\definecolor{shipHull}{RGB}{50, 50, 60}
\definecolor{dangerRed}{RGB}{117,0,20}
\definecolor{foamWhite}{RGB}{240, 250, 255}
 \begin{tikzpicture}[scale=0.8, transform shape, >=Latex]
    % --- 1. Water Background (Gradient) ---
    \shade[top color=oceanBlue!30, bottom color=deepWater!40] (-5, -3) rectangle (6, 3);
        % Faint grid for "Reference Frame" feel
    \draw[white!20, thin, step=1] (-5, -3) grid (6, 3);
    % --- 2. Hydrodynamic Wake (Transparency for realism) ---

    % --- 3. The Ship ---
    \begin{scope}[rotate=0] % Can change rotation here if needed
           % Kelvin wake pattern simulation
    \fill[foamWhite, opacity=0.4] (3.8, 0) -- (-5, 2.5) -- (-3, 0) -- cycle;
    \fill[foamWhite, opacity=0.4] (3.8, 0) -- (-5, -2.5) -- (-3, 0) -- cycle;
    
    % Propeller wash
    \fill[foamWhite, opacity=0.6] (-2.5, 0) -- (-5, 0.5) -- (-5, -0.5) -- cycle;
        % Rudder (Under the hull, so drawn first)
        % Rotated by delta
        \begin{scope}[shift={(-2.8, 0)}, rotate=20] % Rudder angle delta
            \draw[fill=dangerRed!80, thick] (0,0) -- (-0.8, 0.15) -- (-0.8, -0.15) -- cycle;
            \node[below, font=\tiny, color=dangerRed!80] at (-0.5, -0.3) {Rudder $\delta$};
            % Force arrow on rudder
             \draw[<-, thick, dangerRed] (-0.6, 0.2) -- (-0.6, 0.6) node[above, font=\tiny] {$F_{rudder}$};
        \end{scope}

        % Hull Body (Streamlined shape with gradients)
        % Using Bezier curves for realistic shape
        \shade[left color=shipHull!80, right color=shipHull!60, drop shadow={opacity=0.5, shadow xshift=2pt, shadow yshift=-2pt}] 
            (4,0) .. controls (2, 1.2) and (-2, 1.2) .. (-3, 1) 
            -- (-3, -1) 
            .. controls (-2, -1.2) and (2, -1.2) .. (4,0) -- cycle;
            
        % Deck (Inner shape)
        \draw[fill=shipDeck] 
            (3.8,0) .. controls (1.9, 1.0) and (-1.9, 1.0) .. (-2.9, 0.9) 
            -- (-2.9, -0.9) 
            .. controls (-1.9, -1.0) and (1.9, -1.0) .. (3.8,0) -- cycle;

        % Superstructure / Bridge
        \draw[fill=white!90, drop shadow] (-2.2, -0.7) rectangle (-0.8, 0.7);
        \draw[fill=mitLightBlue] (-1.8, -0.5) rectangle (-1.2, 0.5); % Bridge window area
        
        % Center line
        \draw[dashed, opacity=0.5] (-3, 0) -- (4, 0);
        
        \node[font=\bfseries\scriptsize, color=shipHull] at (1, -0.25) {Surface Vessel};

          % Heading Line (Ship long axis extended)
    \draw[dashed, thick, black!70] (0, 0) -- (5.0, 0);
    
    % Heading Angle psi
    \draw[->, thick, black] (1.5, 0) arc (0:30:1.5);
    \node[right] at (1.5, 0.45) {\footnotesize\textbf{$\psi(t)$}};

    % Velocity Vector
    \draw[->, ultra thick, blue!80!black] (4.0, 0) node[below,xshift=0.75cm] {\footnotesize\textbf{U} (Surge Speed)}-- (5.5, 0) ;
    
    % Sway Velocity (v) - usually orthogonal
    \draw[->, thick, blue!60!black] (0, 0) -- (0, -1.0) node[below] {$v$ (Sway)};
    
    % Yaw Rate (r)
  %  \draw[->, thick, orange!80!black] (0.5, 0.8) arc (70:-70:0.4);
   % \node[right, orange!80!black, font=\small] at (0.6, 0.9) {Yaw Rate $r = \dot{\psi}$};

    % Environmental Disturbances (Wind/Waves)
    \draw[->, thick, decorate, decoration={snake, amplitude=1pt, segment length=5pt}, red!60] 
        (-1.0, 2.75) -- (-1, 1) node[midway, left, font=\small] {Wind/Waves $d(t)$};
        \draw[dashed, thick, black!70] (-1.0, 0) -- (2.0, 1.8) node[above] {North / Inertial Frame};
    \end{scope}

    % --- 4. Physics Annotations (Overlaid) ---

    % Global Axis (North)

    % --- 5. Mathematical Model Box ---
    \node[draw=black!50, fill=white, fill opacity=0.9, rounded corners, align=left, drop shadow, anchor=north west] 
        at (1.75, -1.3) {
        \textbf{Norrbin Model ($r=2$)} \\
       % \footnotesize Dynamics: \\ 
        $\dot{\psi} = r$ \\ 
        $T\dot{r} + r + \alpha r^3 = K\delta + d(t)$
    };
    % \node[draw=black!50, fill=white, fill opacity=0.9, rounded corners, align=left, drop shadow, anchor=north west] 
    %     at (-1, -1.5) {
    %    % \textbf{Norrbin Model ($r=2$)} \\
    %    % \footnotesize Dynamics: \\ 
    %     $\dot{x} = U \cos (\psi)$ \\ 
    %     $\dot{y} = U \sin (\psi)$
    % };
\end{tikzpicture}

\caption{Heading control of a surface vessel using the nonlinear Norrbin model. The objective is to track a desired course $\psi_r$ despite wave disturbances.}
\label{fig:ship_control}
\end{figure}

\begin{figure}[htbp]
\centering
\epsfig{figure=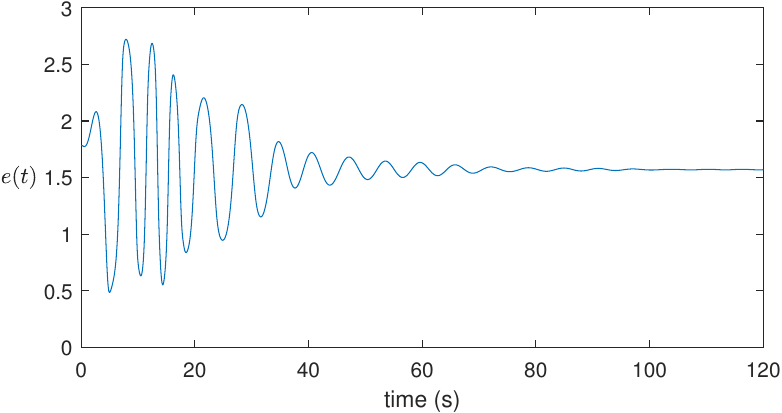,width=0.5\textwidth}
\caption{Tracking performance for the marine surface Vessel ($r=2$).}
\label{fig:vsg_track}
\quad
\epsfig{figure=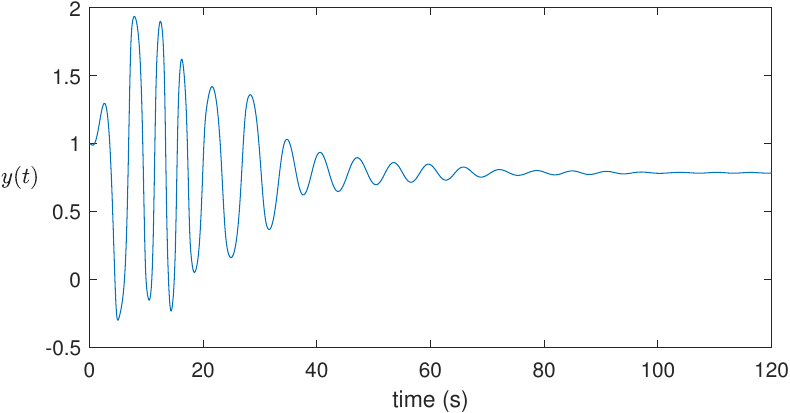,width=0.5\textwidth}
\caption{Desired angle and trajectory of power angle $x_1(t)$.}
\label{fig:Desired_power_angle}

\quad
\epsfig{figure=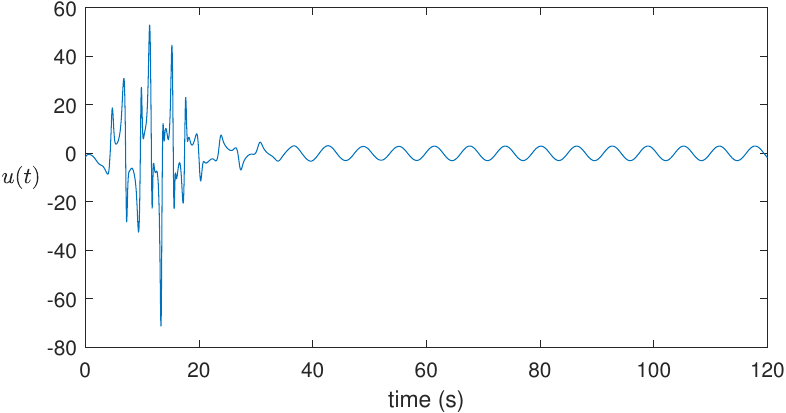,width=0.5\textwidth}
\caption{Control Input $u(t)$ for the marine surface Vessel.}
\label{fig:Control——Input}
\quad
\epsfig{figure=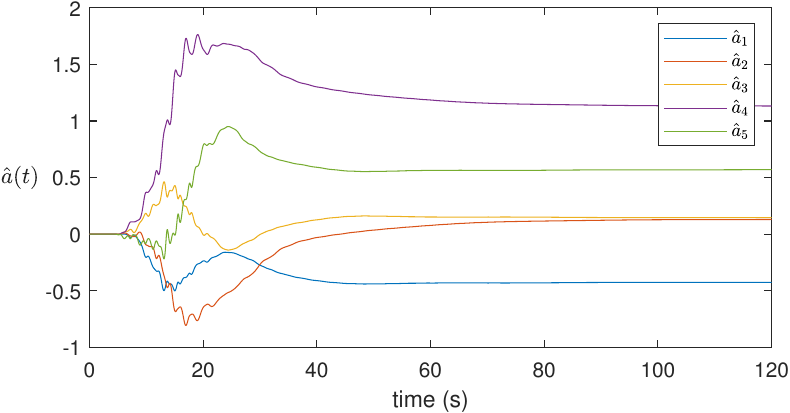,width=0.5\textwidth}
\caption{Parameter estimates for the marine surface vessel.}
\label{fig:Parameter_Estimation}
\end{figure}
\begin{figure}[htbp]
\centering
\epsfig{figure=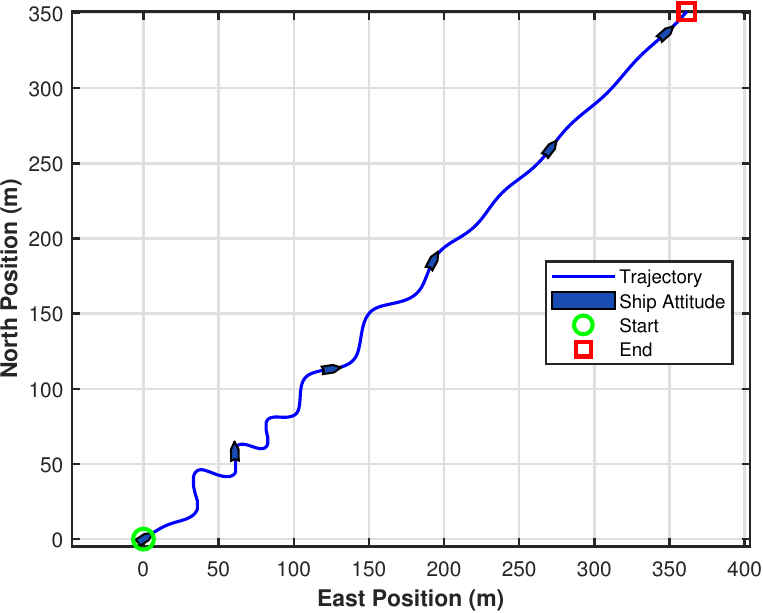,width=0.5\textwidth}
\caption{USV trajectory tracking with constant speed $U = 10$\,m/s by using $\dot{x} = U \cos (\psi)$, $
\dot{y} = U \sin (\psi)$.%the $\left\{\begin{matrix}\dot{x} = U \cos (\psi)\\ \dot{y} = U \sin (\psi)\end{matrix}\right.$.
}
\label{fig:Dynamics_Reconstruction_Inverter}
\end{figure}

Consider the robust course-keeping control problem for a surface vessel. The yaw dynamics are often described by the nonlinear Norrbin model \citep{norrbin1971theory,fossen2021handbook}, which captures the essential maneuvering characteristics, including the nonlinear damping effects:
\begin{subequations}\label{eq:ship_dynamics}
\begin{align}
    \dot{x}_1 &= x_2, \\
    \dot{x}_2 &= -\frac{1}{T} x_2 - \frac{\alpha}{T} x_2^3 + \frac{K}{T} u + d(t), \\
    y &= x_1, \quad e = y - y_r,
\end{align}
\end{subequations}
where $x_1 = \psi$ is the heading angle, $x_2 = r$ is the yaw rate, and $u = \delta$ is the rudder angle. $T$ and $K$ are the time constant and gain indices, respectively. The parameter $\alpha$ represents the coefficient of the nonlinear cubic damping, which is crucial for describing the vessel's turning behavior at high speeds. $d(t)$ denotes the unknown environmental disturbances induced by wind, waves, and ocean currents.

Differentiating the output $y=x_1$ twice yields
\begin{align*}
    \dot{y} &= x_2, \\
    \ddot{y} &= -\frac{1}{T} x_2 - \frac{\alpha}{T} x_2^3 + \frac{K}{T} u + d(t).
\end{align*}
The control input $u$ appears in the second derivative, confirming a relative degree of $r=2$. The nonlinearity $-\frac{\alpha}{T} x_2^3$ depends on the unmeasured state $x_2$ (assuming only heading is measured or to test the observer) and the unknown parameter $\alpha$. The disturbance $d(t)$ acts as a matched uncertainty. The control objective is to steer the vessel to a desired heading $y_r$ while rejecting the wave-induced yaw moments.

In the numerical simulation, the physical parameters of the USV are configured based on a standard prototype model:
\begin{itemize}
    \item Gain index: $K = 0.5$\,s$^{-1}$;
    \item Time constant: $T = 3$\,s;
    \item Nonlinear damping coefficient: $\alpha = 1$\,s$^2$.
\end{itemize}
While fixed values are used for the plant simulation, these parameters are treated as \textit{unknown} by the controller. The controller must learn the inverse of the nonlinear steering dynamics online.

To validate the regulation capability under dynamic sea states, the control objective is to track a %time-varying 
constant heading reference:
\begin{align*}
    y_r(t) =\frac{\pi}{4}.
\end{align*}
Although the reference trajectory is defined, the steady-state control input $u_{ss}$ required to maintain this course is unknown a priori due to the parametric uncertainties and external disturbances from the ocean environment. The interaction dynamics between the vessel and the fluid environment, combined with the disturbance $d(t)$, result in an uncertain steady-state behavior. We assume the steady-state input can be described by a linear internal model of the form:
% \begin{align*}
% \frac{d^{6}\hat{u}}{dt^{6}} + a_6 \frac{d^{5}\hat{u}}{dt^{5}} + \dots + a_2 \frac{d\hat{u}}{dt} + a_1 \hat{u} = 0,
% \end{align*}
\begin{align*}
\frac{d^{5}\hat{u}}{dt^{5}} + a_5 \frac{d^{4}\hat{u}}{dt^{4}} + a_4 \frac{d^{3}\hat{u}}{dt^{3}} + a_2 \frac{d\hat{u}}{dt} + a_1 \hat{u} = 0,
\end{align*}
where $\bm{a} = \col(a_1, \dots, a_5)$ is a vector of unknown constant coefficients.

To rigorously test robustness, the disturbance $d(t)$ is designed as a sinusoidal signal simulating regular wave encounters:
\begin{align*}
    d(t) = 0.5 \sin(t).
\end{align*}
The internal model is expected to compensate for this unknown frequency component ($1\,\text{rad/s}$) to achieve asymptotic tracking.

For the control law \eqref{adnon}, we choose $\rho(e)=1+e^2$ to satisfy the inequality condition. The stabilization parameter is set to $k_a=15$ in \eqref{explicit-mas}. The filter parameters $\lambda_1=1.5$ and $\lambda_2=1.5$ are selected to ensure that the matrix $A$ in \eqref{input-filter} is Hurwitz. To construct the internal model dynamics, the parameters 
%$m_1=10$, $m_2=67.85$, $m_3=229.66$, $m_4=510.11$, $m_5=824.27$, $m_6=1016.99$, $m_7=977.93$, $m_8=736.36$, $m_9=429.02$, $m_{10}=187.62$, $m_{11}=58$, and $m_{12}=11.22$ 
$m_1=10$, $m_2=55.11$, $m_3=151.35$, $m_4=270.59$, $m_5=346.41$, $m_6=329.72$, $m_7=234.84$, $m_8=122.69$, $m_9=44.52$ and $m_{10}=9.95$
are chosen to make the matrix $M$ in \eqref{MNINter} Hurwitz. The simulation starts with initial conditions $x(0)=\col(1,0)$, $\eta(0)=\textbf{0}_{12}$, $\hat{k}(0)=2$, and $\hat{a}(0)=\textbf{0}_6$.

The results in Figs.\  \ref{fig:ship_control}--\ref{fig:Dynamics_Reconstruction_Inverter} confirm that the USV successfully tracks the desired heading reference. Fig. \ref{fig:Dynamics_Reconstruction_Inverter} specifically illustrates the X-Y plane trajectory, where the vessel maintains a constant surge speed $U=10$\,m/s while adjusting its heading $\psi$ according to the control law $\dot{X}=U\cos(\psi), \dot{Y}=U\sin(\psi)$, demonstrating effective path following despite the presence of nonlinear damping and wave disturbances.

\subsection{Regulation of a Repulsive Magnetic Levitation System} \label{RRML}
\begin{figure}[htbp]
\centering

\vspace{-0.5cm}

\epsfig{figure=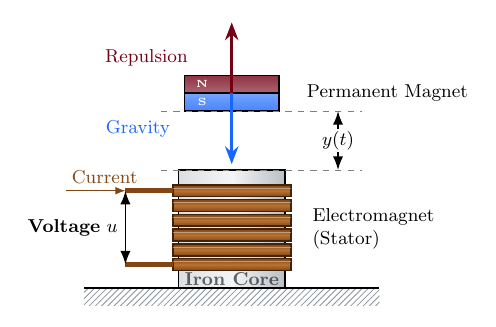,width=0.5\textwidth} 

\vspace{-0.3cm}

\caption{Repulsive magnetic levitation system.}
\end{figure}
To evaluate the proposed strategy for an unstable nonlinear system, consider a repulsive magnetic levitation system, in which a permanent magnet is levitated above an electromagnet coil. The objective is to regulate the levitated height $y$ to a desired trajectory $y_r$. The dynamics are governed by the balance between the electromagnetic repulsion force and gravity, modelled as
\begin{subequations}\label{maglev-dynamics}
\begin{align}
    \dot{x}_1 &= x_2,& \\
    \dot{x}_2 &= \frac{C}{m_g} \frac{x_3}{x_1^2} - g, &\\
    \dot{x}_3 &= -\frac{R}{L} x_3 - \frac{K_b}{L} x_2 + \frac{1}{L} u + d(t),& \\
     y  &=x_1, & e  = y - y_r,
\end{align}
\end{subequations}
where $x_1$ denotes the vertical distance (air gap) between the magnet and the coil, $x_2$ is the vertical velocity, and $x_3$ is the current in the coil. The control input $u$ is the voltage applied to the coil. The system parameters include the mass of the floating magnet $m_g$, gravity $g$, electromagnetic force constant $C$, coil resistance $R$, inductance $L$, and the back-electromotive force coefficient $K_b$. An external disturbance $d(t)$ represents voltage fluctuations or unmodelled electrical dynamics.

Unlike the attractive suspension model, the repulsive magnetic force is modelled as $F_m = C x_3 / x_1^2$, assuming the current $x_3$ is positive. Differentiating the output $y=x_1$ successively until the input appears gives the input-output dynamics:
\begin{equation}\label{eq:maglev_standard}
    y^{(3)} = f(x,t) + \beta(x_1) u.
\end{equation}
The system has a relative degree of $r=3$. The state-dependent high-frequency gain is derived as $\beta(x_1) = \frac{C}{m_g L x_1^2}$, and the nonlinear drift dynamics $f(x,t)$ is given by
\begin{equation}
    f(x,t) = \frac{C}{m_g x_1^2} \left( -\frac{R}{L} x_3 - \frac{K_b}{L} x_2 - \frac{2 x_2 x_3}{x_1} + d(t) \right).
\end{equation}
Both $f(x,t)$ and $\beta(x_1)$ contain strong nonlinearities coupled with the unmeasured states $x_2$ and $x_3$. Specifically, the term $-2 x_2 x_3 / x_1$ introduces a geometric nonlinearity related to the change of magnetic force with position.% To facilitate the control design, we rewrite \eqref{eq:maglev_standard} into the canonical form:
%\begin{equation}\label{eq:canonical_form}
 %   y^{(3)} = \underbrace{f(x,t) +  \beta(x_1) u,
%\end{equation}
%where $u_v=\beta(x_1)^{-1} u$ is the designed control input, and $\sigma(\cdot)$ represents the lumped uncertainty handled by the proposed learning algorithm.
\begin{figure}[htbp]
\centering
  \epsfig{figure=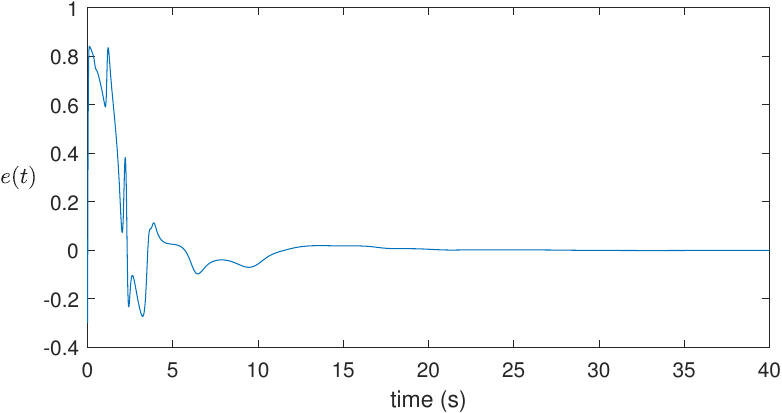,width=0.5\textwidth} 
 \caption{Tracking performance for the repulsive magnetic levitation system ($r=3$ and $e=y(t)-0.5$).}
\label{fig:maglev-tracking}
\quad
\centering
   \epsfig{figure=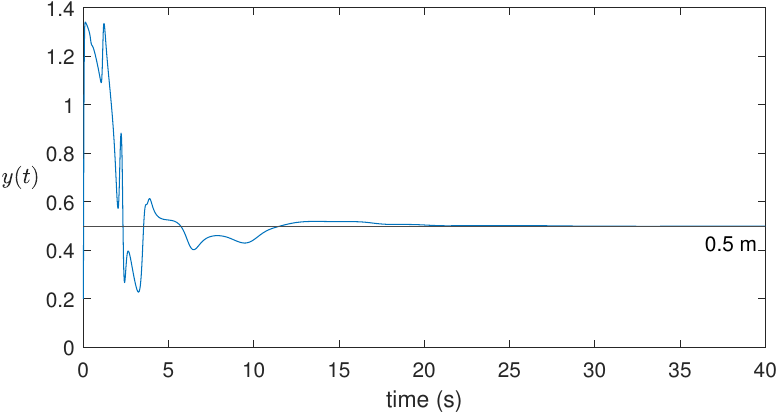,width=0.5\textwidth} 
\caption{Desired vertical distance (air gap) between the
magnet and the coil.}
\label{fig:maglev-error}
\quad
\centering
   \epsfig{figure=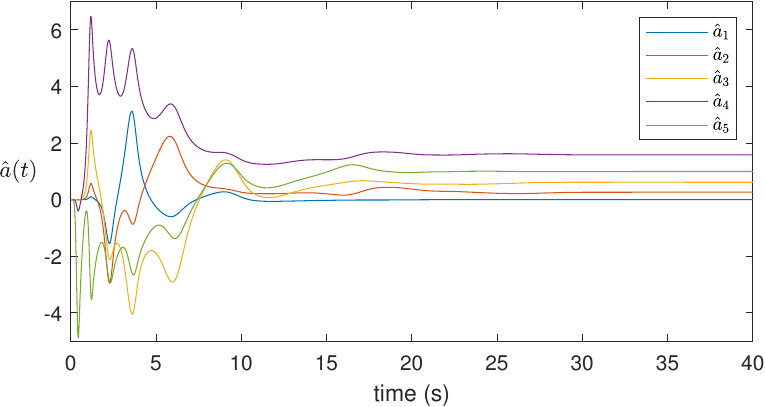,width=0.5\textwidth} 
 \caption{Parameter estimates for the repulsive magnetic levitation system.}
\label{fig:maglev-parameter}

% \quad
% \centering
%    \epsfig{figure=Figure/example-1/DynamicGain_RMLS-eps-converted-to.pdf,width=0.5\textwidth} 
%  \caption{Dynamic gain for the Repulsive Magnetic Levitation System}
% \label{fig:maglev-gain}
\quad
\centering
   \epsfig{figure=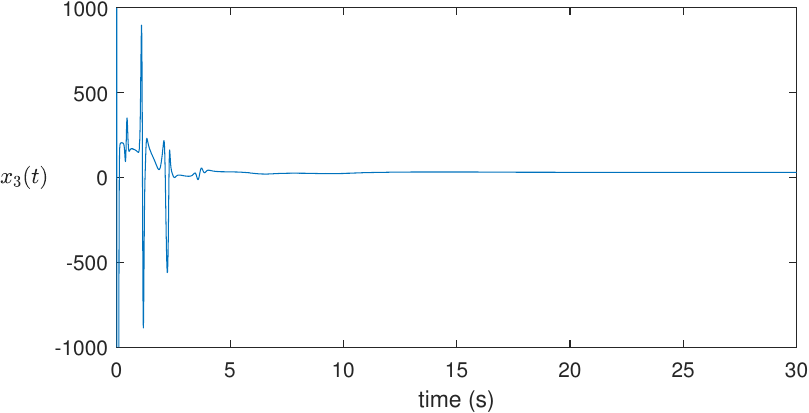,width=0.5\textwidth} 
 \caption{Voltage for the repulsive magnetic levitation system.}
\label{fig:maglev-voltage}
\end{figure}
Since only the position output $y$ is available, the proposed input-driven filter is used to recover the state information. The controller aims to maintain the magnet at a stable levitation height despite the inherent instability of the open-loop dynamics. The complex dynamics of the repulsive magnetic levitation system \eqref{maglev-dynamics} and the unknown exosystem \eqref{eqn: exosystem system} result in an unknown steady-state behavior, making it challenging and impossible to derive an explicit solution, especially considering that only the output is available. Therefore, assume that the system 
\begin{align*}
\frac{d^{5}\bm{\hat{u}}}{dt^{5}}+a_1 \bm{\hat{u}}+a_2 \frac{d\bm{\hat{u}}}{dt}+a_3 \frac{d^{2}\bm{\hat{u}}}{dt^{2}}+a_4\frac{d^{3}\bm{\hat{u}}}{dt^{3}}+a_5\frac{d^{4}\bm{\hat{u}}}{dt^{4}}=0,
\end{align*}
can describe the steady-state input, where $a=\col(a_1,a_2,a_3,a_4,a_5)$ is the unknown constant vector.

In the simulation, the physical parameters are chosen as $m_g=0.1$\,kg, $g=9.8$\,m/s$^2$, $C=1\times 10^{-4}$, $R=2\,\Omega$, $L=0.05$\,H, and $K_b = 0.01$\,V$\cdot$s/m. The reference trajectory is set to $y_r(t) = 0.5$\,m (levitating at a height of $50~cm$). The disturbance is set as $d(t) = 2\sin(0.5t)$. For the control law \eqref{adnon}, we can choose $\rho(e)=20+20e^2$ based on \eqref{rho-inquality-2}  to make the inequality \eqref{Uderivative} negative definite, $k_a=2$ is any positive number in \eqref{explicit-mas}, $\lambda_1=4.5$, $\lambda_2=6.5$ and $\lambda_3=3$ are chosen to make $A$ in \eqref{input-filter} to be Hurwitz. $m_1=1$, $m_2=6.9552$, $m_3=23.6871$,
$m_4=51.8675$, $m_5=80.7073$, $m_6=93.1412$, $m_7=80.7073$, $m_8=51.8675$, $m_9=23.6871$, and $m_{10}=6.9552$ are chosen to make $M$ in \eqref{explicit-mas1} defined in \eqref{MNINter} to be Hurwitz.
The initial position is $x_1(0)=0.2$\,m. The initial velocity is $x_2(0)=0$\,m/s. The initial current in the coil is $x_3(0)=-10$. The initial conditions for the controller are set to $\eta(0)=\textbf{0}_{10}$, $\hat{a}(0)=\textbf{0}_5$, and $\hat{k}(0)=0.5$.
 
The simulation results are shown in Figs.~\ref{fig:maglev-tracking}--\ref{fig:maglev-parameter}. Specifically, Fig.~\ref{fig:maglev-tracking} demonstrates that the tracking error converges to zero, indicating that the proposed method successfully drives the magnet to the target height. 
Fig.~\ref{fig:maglev-error} illustrates the desired vertical distance (air gap) between the magnet and the coil, showing that the magnet is successfully levitated to the target height of 0.5\,m. Finally, 
Fig.~\ref{fig:maglev-parameter} illustrates the evolution of the parameter estimates $\hat{a}(t)$, confirming the adaptation capability of the learning algorithm.

The repulsive magnetic levitation example is representative of a broader class of electromechanical systems with non-affine nonlinear coupling between mechanical states and electrical dynamics. 
While its state-space model does not strictly fall into the high-order normal form \eqref{second-nonlinear-systems} assumed throughout the paper, the proposed controller \eqref{nonpara-control} still achieves the prescribed regulation/tracking objective in this benchmark.
This observation suggests that the design may possess a degree of robustness to structural mismatches beyond the nominal model class. 
A rigorous extension of the analysis to cover such general nonlinear electromechanical systems—where the “next-state affine” structure is not satisfied—will be pursued in future work.

\section{Conclusion}\label{conlu}

We have presented a nonparametric learning solution for the robust output regulation of nonlinear systems via output feedback. 
This framework effectively converts the regulation task into a robust stabilization problem for systems with integral Input-to-State Stable (iISS) inverse dynamics. 
Unlike traditional adaptive methods, our approach eliminates the need for explicit regressor construction and strict Lyapunov function design for parameter learning. 
The approach is illustrated in three numerical examples, involving a repulsive magnetic levitation system, a virtual synchronous generator, and a marine surface vessel induce unknown steady-state chattering, showing convergence of the parameter estimation error and of the tracking error to zero.
It is worth noting that the repulsive magnetic levitation system in Subsection \ref{RRML} is not directly covered by the normal form \eqref{second-nonlinear-systems}, due to its non-affine structural characteristics and complex state-input coupling.
Nevertheless, the proposed control scheme remains effective and meets the objective, demonstrating its broad applicability to generalized nonlinear systems. 
Future work will focus on establishing a unified framework that explicitly accommodates such general electromechanical structures and provides corresponding stability and performance guarantees, as well as on validating the proposed methodology in experiments, including the control of impinging-jet mixers for the manufacturing of solid lipid nanoparticle.

%RDB: great, ready to submit. Thanks a lot!

%\footnotesize
\bibliographystyle{ieeetr}
\bibliography{myref}

\begin{IEEEbiography}[{\includegraphics[width=1in,height=1.25in, clip,keepaspectratio]{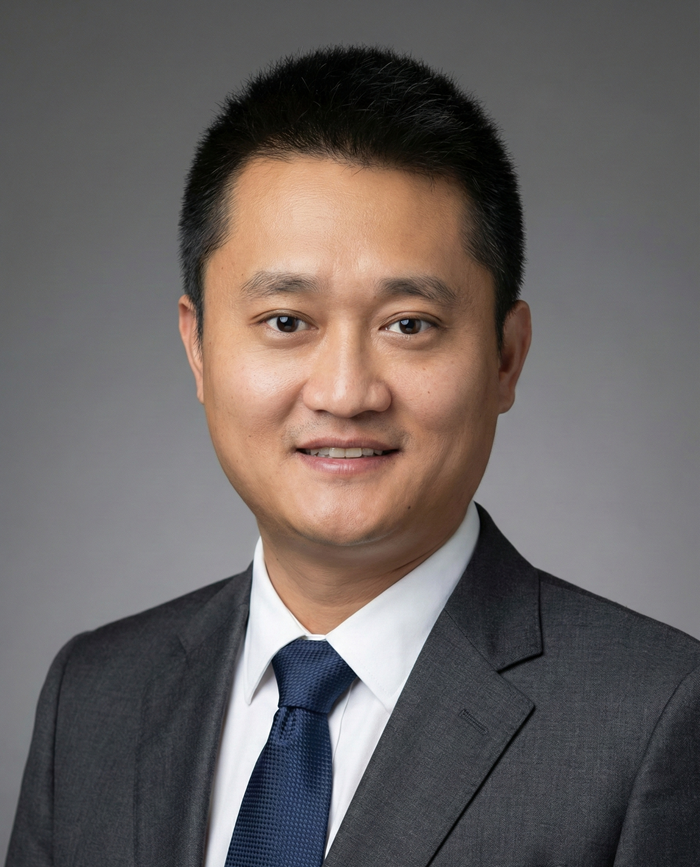}}]{Shimin Wang} 
%received a B.Sci. in Mathematics and Applied Mathematics and an M.Eng. in Control Science and Control Engineering from Harbin Engineering University in 2011 and 2014, respectively. He then 
received a Ph.D. from The Chinese University of Hong Kong. 
He held postdoctoral positions at the University of Alberta, Queen's University and Massachusetts Institute of Technology. His research interests include control engineering and applied mathematics with applications to advanced manufacturing systems. He is the recipient of the Best Conference Paper Awards at the IEEE International Conference on Information and Automation in 2018 and IEEE International Conference on Unmanned Systems in 2025, respectively.  He has also received numerous other honors, including the Best Poster Paper Award at the Nonlinear System and Control Conference in 2024, the MIT Kaufman Teaching Certificate in 2024, and the NSERC Post-Doctoral Fellow Award in 2022. 
\end{IEEEbiography}

\begin{IEEEbiography}[{\includegraphics[height=1.25in, clip,keepaspectratio]{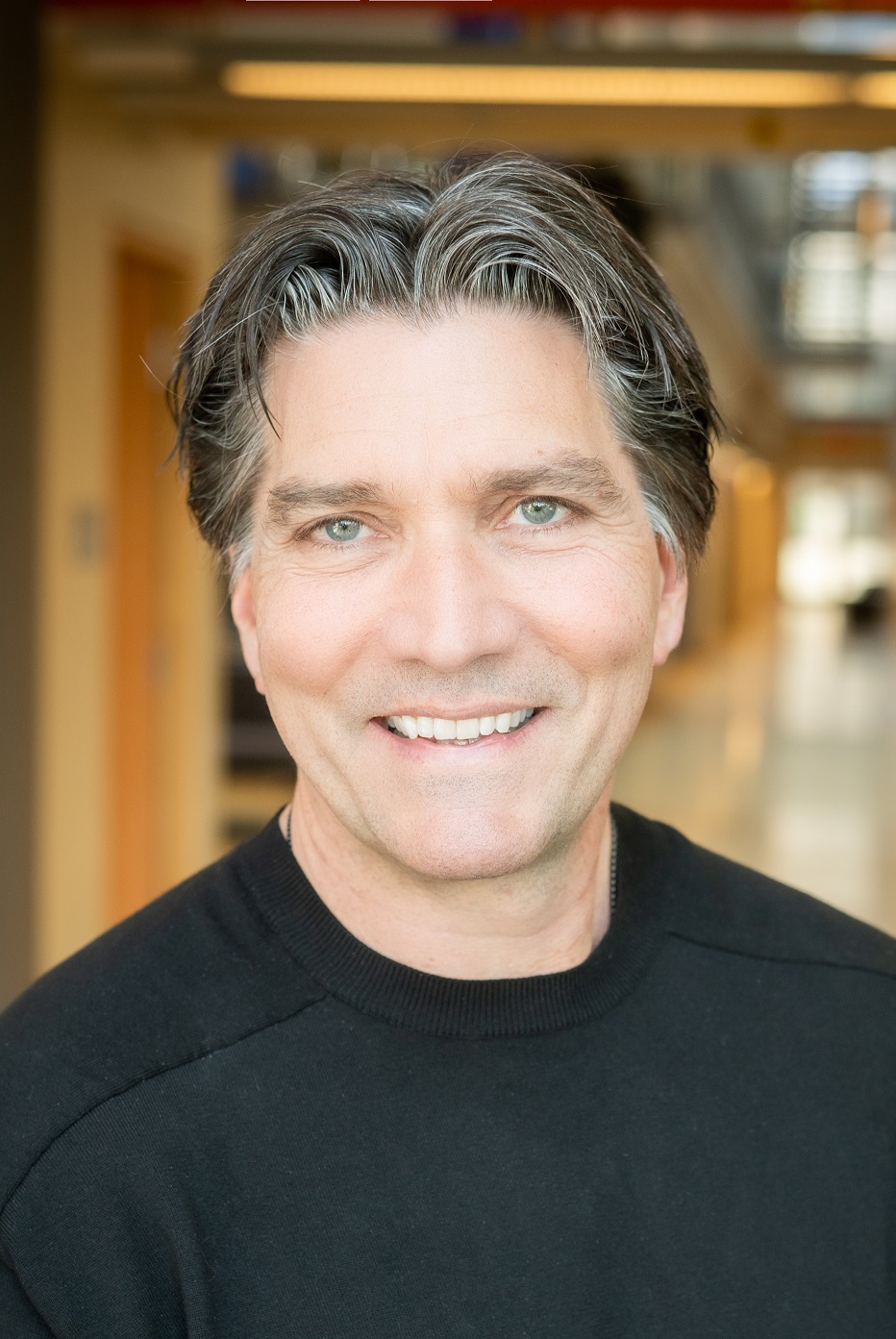}}]{Martin Guay} received a Ph.D. from Queen’s University, Kingston, ON, Canada in 1996. He is currently a Professor in the Department of Chemical Engineering at Queen’s University. His current research interests include nonlinear control systems, especially extremum-seeking control, nonlinear model predictive control, adaptive estimation and control, and geometric control. 

He was a recipient of the Syncrude Innovation Award, the D. G. Fisher from the Canadian Society of Chemical Engineers, and the Premier Research Excellence Award. He is a Senior Editor of IEEE Control Systems Letters. He is the Editor-in-Chief of the Journal of Process Control. He is also an Associate Editor for Automatica, IEEE Transactions on Automatic Control and the Canadian Journal of Chemical Engineering.%, and Nonlinear Analysis: Hybrid Systems.
  \end{IEEEbiography}

\begin{IEEEbiography}[{\includegraphics[width=1in,height=1.25in, clip,keepaspectratio]{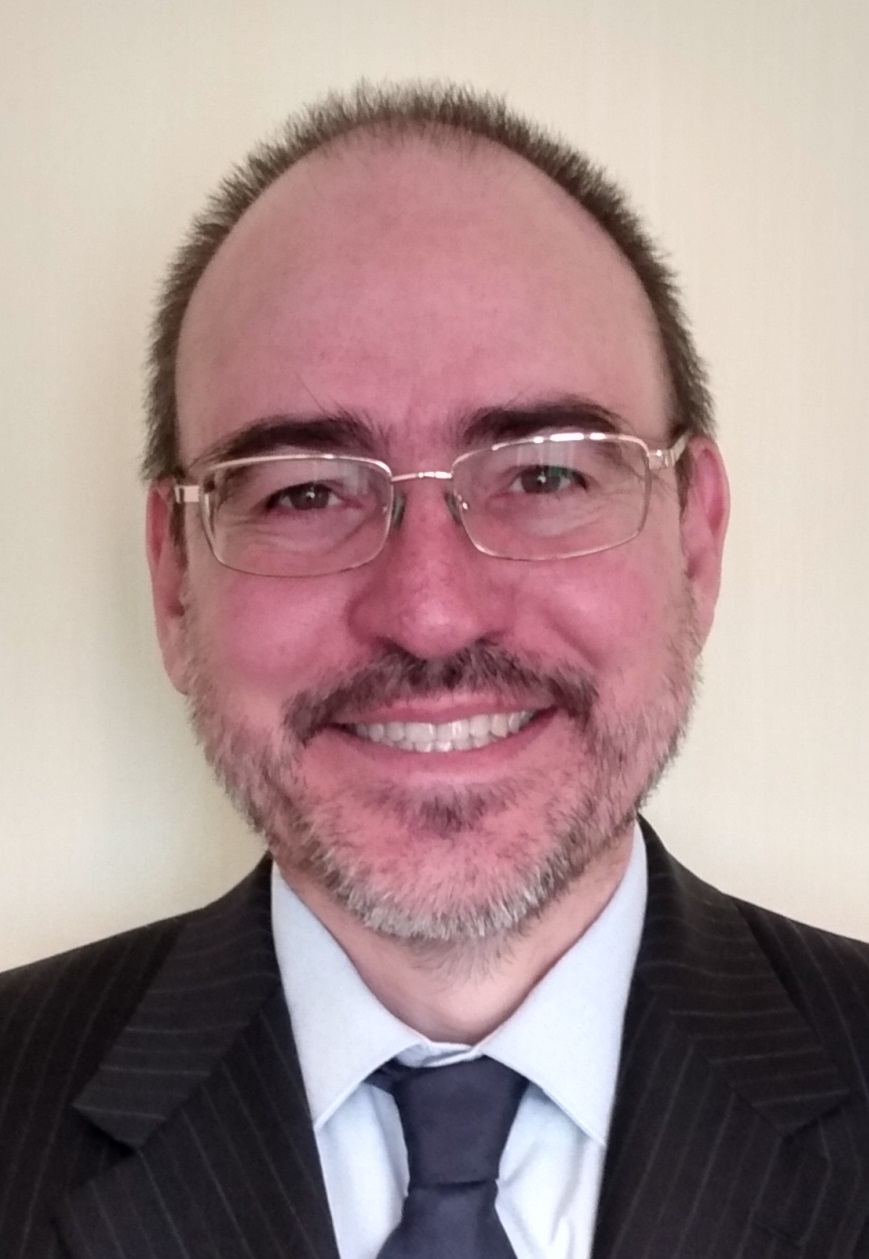}}]{Richard D. Braatz} is the Edwin R. Gilliland Professor at the Massachusetts Institute of Technology (MIT) where he does research in applied mathematics and robust control theory with applications to advanced manufacturing systems. He received an M.S. and Ph.D. from the California Institute of Technology and was on the faculty at the University of Illinois at Urbana–Champaign and was a Visiting Scholar at Harvard University before moving to MIT. He is a past Editor-in-Chief of IEEE Control Systems and a past President of the American Automatic Control Council. Honors include the AACC Donald P. Eckman Award, the Curtis W. McGraw Research Award from the Engineering Research Council, the Antonio Ruberti Young Researcher Prize, and best paper awards from IEEE- and IFAC-sponsored control journals. He is a member of the U.S. National Academy of Engineering and a Fellow of IEEE and IFAC.
\end{IEEEbiography}
\end{document}